\documentclass[twocolumn, amsthm]{autart}

\usepackage[T1]{fontenc}
\usepackage[utf8]{inputenc}
\usepackage{setspace}
\usepackage{xspace}
\usepackage{graphicx}
\graphicspath{ {./figures/} }
\usepackage{subcaption}
\usepackage{amsmath,amssymb}
\usepackage{mathtools}
\usepackage{booktabs}
\usepackage{multirow}
\usepackage{url}
\usepackage{tabularx}

\usepackage{algorithm}
\usepackage{algpseudocode}
\usepackage[acronym]{glossaries}
\newacronym{pwa}{PWA}{piece-wise affine}
\newacronym{lp}{LP}{Linear Programming}
\newacronym{nndm}{NNDM}{Neural Network Dynamical Model}
\newacronym{sbf}{SBF}{Stochastic Barrier Function}
\newacronym{mbr}{MBR}{Minimum Bounding Rectangle}
\newacronym{saa}{SAA}{Sample Average Approximation}
\newacronym{ltl}{LTL}{Linear Temporal Logic}
\newacronym{pac}{PAC}{Probably Approximately Correct}
\newacronym{iid}{i.i.d.}{independent and identically distributed}
\newacronym{sos}{SoS}{sum-of-squares}
\newacronym{lbp}{LBP}{Linear Bound Propagation}

\newtheorem{defi}{Definition}

\newtheorem{assump}{Assumption}
\newtheorem{lemma}{Lemma}
\newtheorem{theorem}{Theorem}
\newtheorem{corollary}{Corollary}
\newtheorem{remark}{Remark}

\newcommand{\uncertaintyspace}{\ensuremath{\Omega}}
\newcommand{\uncertaintyelement}{\ensuremath{\omega}}
\newcommand{\sigmaalgebra}{\ensuremath{\mathcal{F}}}
\newcommand{\probmeas}{\ensuremath{\mathbb{P}}}
\newcommand{\noise}{\ensuremath{\eta}}

\newcommand{\sampleset}{{\ensuremath D}}

\newcommand{\barrierl}[0]{{\ensuremath u}}
\newcommand{\barrierc}[0]{{\ensuremath v}}
\newcommand{\cmartingale}[0]{{\ensuremath c}}
\newcommand{\buffervar}[0]{{\ensuremath \xi}}
\newcommand{\region}[0]{{\ensuremath Q}}
\newcommand{\regionset}[0]{{\ensuremath \mathcal{Q}}}
\newcommand{\numberregions}{\ensuremath{\ell}}
\newcommand{\statespace}{{\ensuremath X}}
\newcommand{\safeset}[0]{{\ensuremath X_s}}
\newcommand{\unsafeset}[0]{{\ensuremath X_u}}
\newcommand{\initialset}[0]{{\ensuremath X_0}}
\newcommand{\timehorizon}[0]{{\ensuremath K}}

\DeclareMathOperator*{\argmin}{arg\,min}
\DeclareMathOperator*{\minimise}{\min}
\DeclareMathOperator*{\maximise}{\max}
\DeclareMathOperator*{\conv}{conv}

\DeclareMathOperator{\subjectto}{s.t.}
\newcommand{\dualvariable}{{\ensuremath \lambda}}

\newcommand{\vertices}{\ensuremath{\mathrm{vert}}}
\newcommand{\polyhedron}{\ensuremath{P}}

\DeclareMathOperator{\im}{im}

\makeatletter
\newsavebox\myboxA
\newsavebox\myboxB
\newlength\mylenA

\newcommand*\xoverline[2][0.75]{%
    \sbox{\myboxA}{$\m@th#2$}%
    \setbox\myboxB\null
    \ht\myboxB=\ht\myboxA%
    \dp\myboxB=\dp\myboxA%
    \wd\myboxB=#1\wd\myboxA
    \sbox\myboxB{$\m@th\overline{\copy\myboxB}$}
    \setlength\mylenA{\the\wd\myboxA}
    \addtolength\mylenA{-\the\wd\myboxB}%
    \ifdim\wd\myboxB<\wd\myboxA%
       \rlap{\hskip 0.5\mylenA\usebox\myboxB}{\usebox\myboxA}%
    \else
        \hskip -0.5\mylenA\rlap{\usebox\myboxA}{\hskip 0.5\mylenA\usebox\myboxB}%
    \fi}
\makeatother

\begin{document}

\begin{frontmatter}
\title{A data-driven approach for safety quantification of non-linear stochastic systems with unknown additive noise distribution}

\author[dcsc]{Frederik Baymler Mathiesen}\ead{f.b.mathiesen@tudelft.nl},
\author[oxford]{Licio Romao},
\author[transportdelft]{Simeon C. Calvert},
\author[dcsc]{Luca Laurenti}, and
\author[oxford]{Alessandro Abate}

\address[dcsc]{Delft Center for Systems and Control, TU Delft.}
\address[oxford]{Department of Computer Science, University of Oxford.}
\address[transportdelft]{Department of Transport \& Planning, TU Delft.}

\thanks[footnoteinfo]{Frederik Baymler Mathiesen and Licio Romao contributed equally to this work. Corresponding author is Frederik Baymler Mathiesen.}

\date{}

\maketitle


\begin{abstract}
In this paper, we present a novel data-driven approach to quantify safety for non-linear, discrete-time stochastic systems with unknown noise distribution. We define safety as the probability that the system remains in a given region of the state space for a given time horizon and, to quantify it, we present an approach based on Stochastic Barrier Functions (SBFs). In particular, we introduce an inner approximation of the stochastic program to design a SBF in terms of a chance-constrained optimisation problem, which allows us to leverage the scenario approach theory to design a SBF from samples of the system with Probably Approximately Correct (PAC) guarantees. 
Our approach leads to tractable, robust linear programs, which enable us to assert safety for non-linear models that were otherwise deemed infeasible with existing methods. To further mitigate the computational complexity of our approach, we exploit the structure of the system dynamics and rely on spatial data structures to accelerate the construction and solution of the underlying optimisation problem. We show the efficacy and validity of our framework in several benchmarks, showing that our approach can obtain substantially tighter certificates compared to state-of-the-art with a confidence that is several orders of magnitude higher.
\end{abstract}

\end{frontmatter}

\section{Introduction}
\label{sec:introduction}

Safety-critical applications, such as autonomous driving \cite{Shalev-Shwartz2017} and robotics \cite{Livingston2012} require provable guarantees of safety, as undesirable behaviours may lead to catastrophic outcomes with long-term economic costs. As a consequence, asserting the safety of complex non-linear noisy systems has been the focus of many recent approaches \cite{PJP:07}, \cite{Aba:17}, \cite{Santoyo2021}, \cite{cosner2023robust}, \cite{Lavaei2021}, \cite{Huijgevoort2023b}. 
However, these approaches generally suffer from exponential complexity with respect to the dimension of the state space. Besides, generally, these approaches require that the distribution of the noise affecting the system is either known and Gaussian or of bounded support \cite{Santoyo2021}, \cite{PJP:07}. Unfortunately, in practice, the noise characteristics of the system are often not known \cite{Aba:17}, \cite{badings2023robust}, \cite{SALAMATI20217}. This leads to the main question of this paper: how can we compute formal certificates of safety for non-linear stochastic systems with unknown noise distribution?

Following the existing literature \cite{Santoyo2021}, \cite{steinhardt2012finite}, \cite{Abate2008}, we define safety as the probability that the system will remain within a given safe set for a given time horizon. 
Common approaches to quantify safety for non-linear stochastic systems either rely on formal abstraction methods \cite{Badings2022}, \cite{cauchi2019efficiency} or on \glspl{sbf} \cite{Santoyo2021}, \cite{steinhardt2012finite}, \cite{Mathiesen2013}. Abstraction-based methods build a discrete representation (i.e., a variant of a Markov chain) of the underlying stochastic system via the discretisation of its state space \cite{Badings2022}, \cite{cauchi2019efficiency}. Then, value iteration is performed on such representation to verify properties and synthesise controllers, which can be mapped back on the original system by relying on simulation relations between the discrete representation and the underlying system \cite{9144391}, \cite{GP:07}.
Various approaches have been proposed that combine abstractions and data-driven methods, including distributionally robust methods \cite{Gracia2022}, the scenario approach \cite{Badings2022}, \cite{badings2023robust}, Chernoff bounds \cite{Abolfazl2023}, and Gaussian process regression \cite{Jackson2020}.
However, a common drawback of these approaches is the need to partition (or discretise) the state space of the original stochastic system and to solve the value iteration on the resulting discrete abstraction, which leads to the well-known state-space explosion problem \cite{cauchi2019efficiency}.

In contrast, \glspl{sbf} \cite{steinhardt2012finite}, 
\cite{kushner1967stochastic}, \cite{mcallister2020stochastic} are Lyapunov-like functions, whose level sets allow one to bound the probability that a dynamical system will remain safe for a given time horizon, without the need to explicitly evolve the system over time \cite{Santoyo2021}, \cite{steinhardt2012finite}, \cite{Jagtap2020}. By not requiring an analytical solution to the system's governing equations over time, \glspl{sbf} represent a promising technique to efficiently quantify the safety of stochastic systems.
However, one of the main challenges in \gls{sbf} design, as we are interested in quantitative (rather than binary) safety evaluation, is in finding a \gls{sbf} that does not lead to overly conservative results \cite{laurenti2023unifying}, \cite{cosner2024bounding}. In fact, designing a \gls{sbf} requires the solution to a stochastic optimisation problem whose nature depends on the class of dynamics under study. 
In literature, synthesis of \glspl{sbf} is usually performed with convex optimisation, in particular \gls{sos} optimisation \cite{PJP:07}, \cite{Jagtap2020}, \cite{Santoyo2021} and \gls{lp} for piece-wise constant \glspl{sbf} \cite{mazouz2024piecewise}, or with deep learning \cite{Mathiesen2013}. However, these papers make restrictive assumptions on the noise distribution and on the system dynamics, e.g. Gaussian noise and linear or polynomial dynamics, are often enforced \cite{Gracia2022}, \cite{rahimian2019distributionally}. 
More recently, some papers relax the assumption on the noise distribution by means of data-driven approaches \cite{SALAMATI20217}, \cite{Salamati2023}, which in addition to the level of safety also induces a (formally quantified) confidence. In particular, \cite{SALAMATI20217} uses \gls{saa} to synthesise a \gls{sbf}. This alternative approach, however, suffers from a sample complexity that is linear in the inverse of the confidence. 
Recent work on synthesising \glspl{sbf} purely from trajectory data \cite{schon2024datadriven}, i.e., without assuming partial knowledge of the dynamics, has used Conditional Mean Embeddings, whose bottleneck is a computational complexity of $O(N^3)$ where $N$ is the number of samples. 

Our approach departs from previous techniques to ensure safety for non-linear stochastic dynamical systems using \glspl{sbf}. First, we present an inner approximation to the stochastic program commonly used to design \gls{sbf} in terms of a chance-constrained optimisation problem. The feasible set of the latter is contained in the feasible set of the former.  
Then, by restricting to the class of piece-wise affine stochastic barrier functions\footnote{It is known that with sufficiently many pieces a piece-wise affine function can formally approximate any continuous function arbitrarily well.} and relying on uncertain linear relaxations of non-linear systems \cite{kaidi2020automatic}, we show that the resulting chance-constrained problem can be reformulated into a robust \gls{lp} problem \cite{BV:04}.
This reformulation allows us to employ the scenario approach theory to devise a data-driven framework to synthesise a \gls{sbf}, and consequently obtain safety guarantees to the trajectories of the system. The resulting approach is data-efficient, as it only requires a limited amount of samples from the noise distribution that is logarithmic in the negative inverse of the confidence, which is in contrast with \cite{SALAMATI20217} where, as previously mentioned, the sample complexity is instead proportional in the inverse of the confidence. Furthermore, our approach is scalable due to its \gls{lp} representation. We also introduce an a priori sample discarding procedure and spatial indexing for efficient model construction to improve scalability of the proposed framework.
Our experiments show competitive performance on various systems, including a model of a vehicle in windy conditions and various \glspl{nndm} \cite{mazouzsafety2022}. Our numerical analysis illustrates how our approach substantially outperforms state-of-the-art comparable methods in terms of both the tightness of bounds and the amount of data required. Overall, our main contributions are:

\begin{itemize}
    \item We develop a data-driven technique for the design of barrier functions that relies on a chance-constrained inner approximation of  stochastic programs.
    \item We use the scenario approach theory and \acrlong{lbp} techniques to synthesize a \gls{sbf} for non-linear dynamical systems.
    \item We present an efficient computational architecture to construct the resulting optimisation problem using spatial indexing methods, such as R-trees, for faster set intersection computations.
    \item We benchmark the proposed framework, showing its advantages with respect to approaches in the literature, including instances of \acrfullpl{nndm}.
\end{itemize}

A conference version of this paper appeared in \cite{mathiesen2023inner}. This paper significantly expands the preliminary results in \cite{mathiesen2023inner}, which focused on developing data-driven \gls{sbf} synthesis techniques only for \gls{pwa} dynamics. In this paper, we deal with general non-linear dynamics, which necessitates non-trivial extensions of the techniques in \cite{mathiesen2023inner} to handle uncertain \gls{pwa} models. Furthermore, here we show how to leverage spatial indexing and sample discarding to enhance scalability of our framework. We also newly include detailed proofs, as well as a substantially extended empirical analysis, including experiments on \glspl{nndm}. 

The structure of this paper is as follows. The problem statement is described in Section \ref{sec:pro_formulation}. Section \ref{sec:prelim} covers preliminary results from convex analysis, scenario optimisation, and \acrlong{pwa} relaxations of non-linear functions.
Section \ref{sec:stochastic_barrier_functions} defines piece-wise linear \glspl{sbf} and how they can be used to certify safety for non-linear stochastic systems. Section \ref{sec:chance-constrained-apprx} details the main theoretical results, including a reformulation of the stochastic program to synthesise a \gls{sbf} in to an inner chance-constrained approximation.
In Section \ref{sec:data_driven_uncertain_stochastic_barriers}, we apply scenario approach theory towards a data-driven synthesis of a \gls{pwa} \gls{sbf}. 
In Section \ref{sec:computational_methods_and_efficiency} we show how to compute the necessary polyhedral over-approximations and two methods to improve the scalability, namely the convex hull over sample set and spatial indexing.
Empirical studies are reported in Section \ref{sec:experiments}.

\section{Problem formulation}
\label{sec:pro_formulation}

\subsection{Notation}
Let $\mathbb{R}$, $\mathbb{R}_{\geq 0}$, and $\mathbb{N}$ represent, respectively, the set of real, non-negative real, and natural numbers. We denote by $\{ x_1, \ldots, x_m\}$ the elements in a finite set. Let $\Omega$ be an abstract space, $\mathcal{F}$ be a $\sigma$-algebra defined on this set, and $\mathbb{P}$ be a probability measure; we denote by $(\Omega,\mathcal{F},\mathbb{P})$ the associated complete probability (or uncertainty) space. A random $\noise$ variable with values in $\mathbb{R}^n$ is a measurable function $\noise: \Omega \mapsto \mathbb{R}^n$, with $\mathbb{R}^n$ equipped with its standard Borel $\sigma$-algebra. A realisation of $\eta$ is denoted by $\eta(\omega)$, for some $\omega \in \Omega.$ 
For a given set $\Omega$,  we denote by $\mathcal{P}(\Omega)$ the set of probability measures defined over $\Omega$. 
The supremum norm $\lVert \cdot \rVert_{\infty}$ on a function $f : \mathbb{R}^n \mapsto \mathbb{R}$ is defined as $\lVert f \rVert_{\infty} = \sup_{x \in \mathbb{R}^n} \lvert f(x) \rvert$.

\subsection{Problem formulation}

We consider the following discrete-time, stochastic non-linear system
\begin{equation}
    \label{eq:system}
     {x}(k+1) = f( {x}(k)) + \noise, \quad x(0) = \bar{x},
\end{equation}
where $k \in \mathbb{N}$ denotes time, $\bar{x} \in \mathbb{R}^n$ is the initial condition, and $f : \mathbb{R}^n \mapsto \mathbb{R}^n$ is a Lipschitz continuous function. $\noise$ is a random variable
defined on the probability space $(\Omega, \mathcal{F},\mathbb{P})$ and, at each time instance, an independent realization is drawn and added to the nominal dynamics represented by the function $f$ (\gls{iid}).  We further assume that the probability distribution of $\noise$ is absolutely continuous with respect to the Lebesgue measure of $\mathbb{R}^n$ for a well-defined probability density function.
Throughout this paper, we assume that the probability distribution of $\eta$ is unknown.

Because of the \gls{iid} assumption on $\eta$, we can alternatively write the dynamics of System \eqref{eq:system} using its kernel representation. A \emph{stochastic kernel} is a measurable map from $\mathbb{R}^n$ onto the space of probability measures $\mathcal{P}(\mathbb{R}^n)$, $T: \mathbb{R}^n \mapsto \mathcal{P}(\mathbb{R}^n)$ \cite{bertsekas2004stochastic}. In particular, the stochastic kernel associated with System \eqref{eq:system} is given by
\begin{equation}
    T(X \mid x) = \int_{\Omega} \mathbf{1}_{X}(f(x) + \noise(\omega)) \mathbb{P}(d \omega),
    \label{eq:system-kernel-representation}
\end{equation}
where $X \subset \mathbb{R}^n$ is a Borel set of $\mathbb{R}^n$, and $ \mathbf{1}_{X}$ is the indicator function, i.e., $ \mathbf{1}_{X}(x)=1$ if $x\in X$, and $0$ otherwise. In other words, for a fixed $x\in\mathbb{R}^n$ representing the current state, the stochastic kernel associated with System \eqref{eq:system} returns the probability distribution of the state in the next time step.
For $K \in \mathbb{N}$ and $\bar{x} \in \mathbb{R}^n$, we denote by $(\mathbb{R}^n)^K = \mathbb{R}^n \times \ldots \times \mathbb{R}^n$ the $K$-fold Cartesian product of $\mathbb{R}^n$. Then, the stochastic kernel \eqref{eq:system-kernel-representation} induces a unique measure on $(\mathbb{R}^n)^K$ given by the unique extension (due to Kolmogorov's extension theorem \cite[Theorem 2.4.3]{tao2013introduction}) of the measure
{\small\vspace{-2.8em}\begin{equation*}
    \mathbb{P}^{\bar{x}}(X_1,\ldots,X_K) = \int_{X_1} \left(\prod_{k = 2}^K \int_{X_k} T(d\xi_k|\xi_{k-1}) \right)  T(d\xi_1\mid \bar{x}),
\end{equation*}}
which represents the probability of a trajectory starting at $\bar{x}$, under the dynamics of System \eqref{eq:system}, being in the set $X_k$, $k = 1, \ldots, K$ at the various time steps. 
We now have all the ingredients to define the notion of probabilistic safety that is crucial to our approach.

\begin{defi}[Probabilistic safety \cite{Abate2008}]\label{defi:weak_FI}
Let $K$ be a non-negative integer and $\safeset \subseteq \statespace$ be a compact set representing the safe set. Then, for a given initial state $\bar{x}$, we define probabilistic safety as
\begin{align}
    \zeta(\safeset, K \mid \bar{x}) &= \mathbb{P}^{\bar{x}}(X_s,\ldots,X_s)  \nonumber.
\end{align}
\end{defi}

Our main objective in this paper is to obtain a uniform lower bound on the probabilistic safety of System \eqref{eq:system} for all initial conditions in a given set $\initialset$.

\begin{prob}
\label{prob:MainProblem}
Let $D=(\omega_1,...,\omega_N)$ be \gls{iid} samples from $\eta.$ Then, for a given compact set $\initialset \subseteq \statespace$, a safe set $\safeset \subseteq \statespace$, and a time horizon $K \in \mathbb{N}$, find $\rho \in \left(0,1\right]$ such that, with high confidence,
\[ \zeta(\safeset, K) = \inf_{\bar{x} \in \initialset} \zeta(\safeset, K \mid \bar{x}) \geq \rho. \]
\end{prob}
Note that, as the noise is additive, the assumption of having \gls{iid} noise samples in Problem \ref{prob:MainProblem} is equivalent to having \gls{iid} full measurements of the state.
Exact computation of probabilistic safety for System \eqref{eq:system} is generally infeasible, even when $\noise$ is restricted to being a Gaussian random variable \cite{cauchi2019efficiency}. Consequently, it is obvious that the more general setting considered in this paper, where the noise distribution is arbitrary and possibly unknown, requires approximations. 
We should also stress that, while in Problem \ref{prob:MainProblem} we focus on safety, the techniques developed in this paper to obtain its solution can also be applied to perform verification of more complex temporal properties. In fact, safety is the dual to, and can be reformulated as, reachability\cite{Abate2008}\footnote{The safety probability is one minus the probability of reaching the unsafe set.}, and complex temporal specifications such as LTLf \cite{10.5555/2540128.2540252} can be reduced to reachability by representing the formula as an automaton and checking reachability to an accepting state of the product system between the automaton and the system \cite{Lavaei2021}, \cite{Jagtap2020}, \cite{9144391}. 

\paragraph*{Approach}
We leverage the availability of data $D$ to synthesise a \acrfull{sbf} that, with high confidence, bounds probabilistic safety. Our approach is summarised in Figure \ref{fig:relationship-between-programs} and is based on a novel reformulation of the optimisation problem to find a \gls{sbf} into a chance-constrained problem (detailed in Section \ref{sec:chance-constrained-apprx}). To deal with the non-linearity of a system, we develop in Section \ref{sec:uncertain_dynamics_safety_and_ccbp} a formal (i.e. with a quantified error) approximation of System \eqref{eq:system} as an uncertain \acrfull{pwa} system. Such an approximation is then used in Section \ref{sec:data_driven_uncertain_stochastic_barriers} to reformulate the chance-constrained problem into a robust linear program, which can be efficiently solved using duality. 

\begin{figure}
    \centering
    \includegraphics[width=\linewidth]{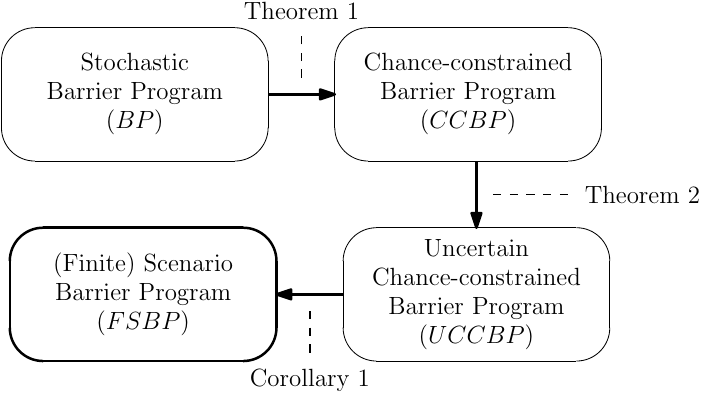}
    \caption{Summary of the proposed framework for safety verification of nonlinear systems. Theorem \ref{theorem:barrier-ccp} (Section \ref{sec:chance-constrained-apprx}) guarantees that the associated stochastic program to find a \gls{sbf} can be solved as a chance-constrained program via constraint tightening. Non-linearities can be dealt with by means of a \gls{pwa} abstraction as in Theorem \ref{theorem:uncertain_barrier-ccp} (Section \ref{sec:uncertain_dynamics_safety_and_ccbp}). The scenario approach is applied to use available data for high-confidence certificates; strong duality is essential to relax the problem to \acrlong{lp} in Corollary \ref{corollary:scenario-barrier-design} (in Section \ref{sec:data_driven_uncertain_stochastic_barriers}).}
    \label{fig:relationship-between-programs}
\end{figure}

\section{Preliminaries}
\label{sec:prelim}
Our approach requires some results from convex optimisation and scenario optimisation, which we summarise below for convenience. 

\subsection{Convex analysis and linear programming}
\label{subsec:prelim_LP}
A class of convex sets that will be widely used in this paper is the class of polyhedral sets \cite{legat2021polyhedra}. Given a matrix $H \in \mathbb{R}^{p \times n}$ and a vector $h \in \mathbb{R}^p$, a polyhedral set is denoted by
\begin{equation}
\polyhedron = \{ x \in \mathbb{R}^n: H x \leq h \}.
\label{eq:polyhedron_set}
\end{equation}
Representation \eqref{eq:polyhedron_set} is called the \emph{half-space representation} of polyhedral sets.
The vertex representation of $\polyhedron$ is $\polyhedron = \conv(x_1,\ldots,x_m)$ where $\vertices(\polyhedron)= \{x_1, \ldots, x_m\}$ are the vertices\footnote{Formally, an element $x$ of a convex set $C$ is called a vertex if whenever $x = \lambda x_1 + (1-\lambda) x_2$ for some $\lambda \in (0,1)$ and $x_1, x_2 \in C$, we have that $x_1 = x_2$.} of $\polyhedron$.

The following class of robust, or semi-infinite, \glspl{lp} is crucial to our developments: 
\begin{equation}
    \begin{aligned}
    \minimise_{z} & \quad c^\top z \\
    \subjectto & \quad (Az + a)^\top x \leq Bz + b, \quad \text{ for all } x \in \polyhedron,
    \end{aligned}
    \label{eq:robust_LP}
\end{equation}
where $z \in \mathbb{R}^d$ is the vector of decision variables of size $d \in \mathbb{N}$, $c \in \mathbb{R}^d$ is the objective cost, $A \in \mathbb{R}^{n \times d}$, $a \in \mathbb{R}^{n}$, $B \in \mathbb{R}^{1 \times d}$, $b \in \mathbb{R}$ are constraint coefficients, and $\polyhedron \subset \mathbb{R}^n$ is a polyhedral set. By relying on standard duality arguments, we can recast the semi-infinite optimisation problem \eqref{eq:robust_LP} as a regular LP as shown in the following proposition.

\begin{prop}[{{\cite[Exercise 5.17]{BV:04}}}]
    Consider the semi-infinite \gls{lp} problem \eqref{eq:robust_LP}, and denote by
     \[
    \mathcal{Z} = \{ z \in \mathbb{R}^d : (Az + a)^\top x \leq Bz + b, \text{ for all } x \in \polyhedron \},
    \]
    its feasible set. Define the optimisation problem
    \begin{align}
        \minimise_{z,\dualvariable} & \quad c^\top z \nonumber \\
        \subjectto & \quad h^\top \dualvariable \leq Bz + b \nonumber \\ 
        & \quad H^\top \dualvariable = Az + a, \quad \dualvariable \geq 0,
        \label{eq:dual_robust_LP_formulation}
    \end{align}
    whose feasible set is given by
    \[
    \mathcal{Z}' = \{ z : \exists \dualvariable \in \mathbb{R}^p_{\geq 0}, ~h^\top \dualvariable \leq Bz + b,~ H^\top \dualvariable = (Az + a) \}.
    \]
    Then, it holds that $\mathcal{Z} = \mathcal{Z}'$.
    \label{prop:main_result_robust_LP}
\end{prop}

\subsection{Scenario optimisation}
\label{subsec:prelim_scenario_theory}

The scenario approach theory allows one to certify the probability of constraint violation associated with chance-constrained optimisation problems \cite{CG:08}. Let $\eta: \Omega \mapsto \mathbb{R}^q$ be a random variable on the probability space $(\Omega,\mathcal{F},\mathbb{P})$. Let $\sampleset = (\omega_1, \ldots, \omega_N)$ be \gls{iid} samples from $\probmeas$, which live naturally in the space $(\uncertaintyspace^N, \otimes_N \sigmaalgebra, \probmeas^N)$, where $\uncertaintyspace^N$ is the $N$-fold Cartesian product of $\uncertaintyspace$, and $\otimes_N \sigmaalgebra$ is the product sigma algebra generated by the sigma algebra $\sigmaalgebra$, and $\probmeas^N$ represents the induced measure on $\uncertaintyspace^N$. Then, consider the scenario program
\begin{align}
    \minimise_{z} & \quad c^\top z \nonumber \\
    \subjectto & \quad g(z, \noise(\omega)) \leq 0, \quad \text{ for all } \omega \in \sampleset,
    \label{eq:scenario_program}
\end{align}
where $d$ is the dimension of the optimisation variables, $c \in \mathbb{R}^d$ is the objective cost, and $g(z,\eta): \mathbb{R}^d \times \mathbb{R}^q$ is a function that is convex in $z$ for each value of $\eta$ and measurable in $\eta$ for each value of $z$. Notice that the scenario program \eqref{eq:scenario_program}, by enforcing one convex constraint per sample in $\sampleset$, is convex program and consequently, it can be solved using convex optimization tools \cite{gearhart2013comparison}, \cite{Huangfu2017}.

\begin{assump}
We assume almost surely with respect to the measure $\mathbb{P}^N$ that:
    \begin{itemize}
        \item[a.] The feasible set $\mathcal{Z} = \{ z : g(z, \noise(\omega)) \leq 0, \forall \omega \in \sampleset \}$ has non-empty interior.
        \item[b.] The optimal solution of program \eqref{eq:scenario_program} exists and is unique.
    \end{itemize}
    \label{assump:well-posedness-scenario}
\end{assump}
Both conditions in Assumption \ref{assump:well-posedness-scenario} are standard. In fact, uniqueness can always be enforced with a tie-break rule.
Throughout this paper, we denote the unique solution of \eqref{eq:scenario_program} by $z^\star(\sampleset)$, which is a well-defined random variable on the space $\uncertaintyspace^N$. A key result within the scenario approach theory establishes an upper bound on the tail distribution of the constraint violation probability associated with $z^\star(\sampleset)$. 

\begin{prop}[\cite{CG:08}]
Consider the scenario program \eqref{eq:scenario_program} and suppose Assumption \ref{assump:well-posedness-scenario} holds. Then, for any $\epsilon \in (0,1)$, we have that
\[
\probmeas^N \{ \sampleset \in \uncertaintyspace^N: V(z^\star(\sampleset)) > \epsilon \} \leq \sum_{i = 0}^{d-1} \binom{N}{i} \epsilon^i (1-\epsilon)^{N-i}, 
\]
where 
$
    V(z) = \probmeas \{ \omega \in \uncertaintyspace : g(z,\noise(\omega)) > 0 \} 
$
is the violation probability of $z \in \mathbb{R}^d$.
\label{prop:scenario_approach_theory}
\end{prop}

The inequality in Proposition \ref{prop:scenario_approach_theory} holds with equality for the class of fully-supported scenario programs -- the reader is referred to \cite{CG:08} for more details. 

\subsection{Uncertain piece-wise affine relaxations}
\label{subsec:prelim_lbp}

Uncertain \gls{pwa} relaxations, as depicted in Fig. \ref{fig:auto_lirpa}, are key to our method. This type of relaxation allows one to treat complex non-linear functions as uncertain \gls{pwa} functions, simplifying analysis and optimisation.
An uncertain \gls{pwa} relaxation for a function $f$ is a collection of local linear relaxations, that is, $\underline{A_i} x + \underline{b_i} \leq f(x) \leq \xoverline{A_i} x + \xoverline{b_i}$ for all $x$ in a convex region $\region_i$, given a partition $\regionset = \{ \region_1, \ldots, \region_\numberregions \}$.
We call a partition convex if each region in the partition is convex.
The union of all regions $\bigcup_{i = 1}^\numberregions \region_i$ is the domain of the relaxation. 
We note that any locally Lipschitz function $f$ can be relaxed to an uncertain \gls{pwa} function \cite{CLARKE198152}. The idea is to partition the input domain into a number of regions and compute linear relaxations independently for each region. We formalise the existence of an uncertain \gls{pwa} relaxation in what follows.
\begin{prop}\label{prop:continuous_function_as_uncertain_pwa}
    Let $\regionset = \{ \region_1, \ldots, \region_\numberregions \}$ be a given convex partition of a compact set $\statespace \subset \mathbb{R}^n$. Then, for any function $f : \mathbb{R}^n \to \mathbb{R}^n$ that is locally Lipschitz on $X$, there exists an uncertain \gls{pwa} function \[\hat{f}(x, \alpha) = \hat{f}_i(x, \alpha) = A_i(\alpha)x + b_i(\alpha), \quad \text{ for } x \in \region_i \subseteq X\] where $\alpha \in [0, 1]$ and
    \[
        A_i(\alpha) = \alpha \underline{A_i} + (1 - \alpha) \xoverline{A_i}, \quad b_i(\alpha) = \alpha \underline{b_i} + (1 - \alpha) \xoverline{b_i}
    \]
    with matrices $\underline{A_i}, \xoverline{A_i}\in \mathbb{R}^{n \times n} $ and vectors $\underline{b_i}, \xoverline{b_i} \in \mathbb{R}^n$ given for all $\region_i \in \regionset$, such that it holds that $f(x) \in \{ \hat{f}(x, \alpha)  : \alpha \in [0, 1] \}$ for all $x \in \statespace$.
\end{prop}
\begin{figure}
    \centering
    \includegraphics[width=\linewidth]{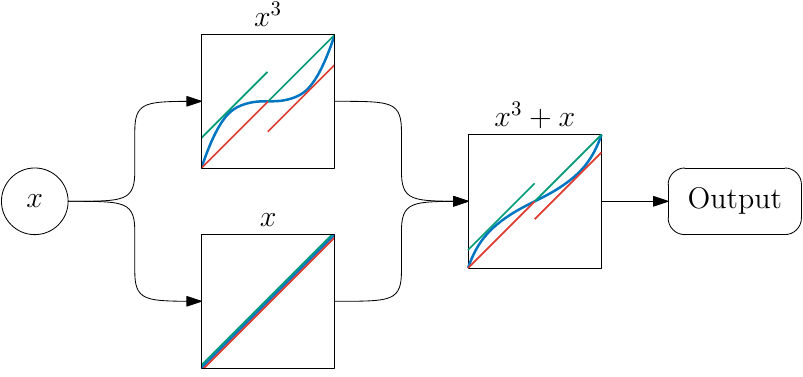}
    \caption{Computation graph for the function $f(x) = x^3 + x$ with \acrfull{lbp} annotation for the input regions $[-1, 0]$ and $[0, 1]$. \Acrshort{lbp} operates by propagating backward linear bounds on the computation graph.} 
    \label{fig:auto_lirpa}
\end{figure}
A proof for the proposition can be found in Appendix \ref{subsec:appendix_proof_existence_linear_bounds}. Note that there are various possible approaches to select 
matrices $\underline{A_i}, \xoverline{A_i} $ and vectors $\underline{b_i}, \xoverline{b_i}$ for each region $\region_i$. In this paper, we use a state-of-the-art technique, called \acrfull{lbp}\glsunset{lbp} \cite{kaidi2020automatic}.
The core idea of \gls{lbp} is to recursively propagate linear bounds backward through the computation graph representing a function.
An example of this bound propagation procedure can be seen in Fig. \ref{fig:auto_lirpa} where the cubic term is linearly bounded based on the input interval bounds and composed with the linear term.

\section{Stochastic barrier functions}\label{sec:stochastic_barrier_functions}

In this paper, we use \acrlongpl{sbf} \cite{PJP:07} to provide safety guarantees for System \eqref{eq:system}.

\begin{defi}[\acrlong{sbf}]\label{defi:Barrier_certificate}
Given a safe set $\safeset$ and a set of initial conditions $\bar{X}$, a measurable function $B: \mathbb{R}^n \mapsto \mathbb{R}_{\geq 0}$ is called a \acrfull{sbf} for System\eqref{eq:system} if there exist non-negative constants $\gamma, \cmartingale$ satisfying
\begin{align}
    B(x) \leq \gamma, &~ \text{ for all } x \in \initialset,\label{eq:initial_set_constraint}\\
    \label{eq:unsafe-set-cond}
    B(x) \geq 1, & ~ \text{ for all }x \in \mathbb{R}^n \setminus \safeset,\\
    \label{eq:c-martingale}
     \mathbb{E} \left\{ B(f(x) + \noise(\uncertaintyelement)) \right\}& 
      \leq B(x) + \cmartingale, ~ \text{ for all } x \in \safeset.
\end{align}
\end{defi}

\begin{figure}
    \centering
    \includegraphics[width=\linewidth]{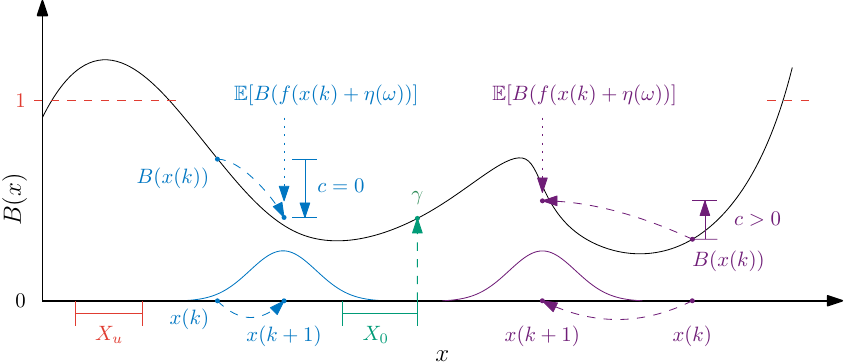}
    \caption{The figure is borrowed from \cite{Mathiesen2013}. A \gls{sbf} $B(x)$ is a non-negative function that is greater than $1$ in an unsafe region $\unsafeset$, which is the complement of the safe set $\safeset$. The variable $\gamma$ is an upper bound for $B(x)$ over an initial region $\initialset$. The upper bound for the expected increase in $B(x)$ after one step of \eqref{eq:system} over the safe set $\safeset$ is denoted $\cmartingale$. Proposition \ref{prop:barrier_prob_safety} shows that $\zeta(\safeset, T) \geq 1 - (\gamma + \cmartingale T)$.}
    \label{fig:example_continuous_barrier_function}
\end{figure}

A pictorial representation of a \gls{sbf} is presented in Figure \ref{fig:example_continuous_barrier_function}. Inequality \eqref{eq:c-martingale} requires that for all $x\in \safeset$, the expected value of the barrier function at the next step cannot increase more than $c$. As shown in Proposition \ref{prop:barrier_prob_safety} below, we can leverage results from martingale theory to obtain a lower bound on $\zeta(\safeset, \timehorizon)$.  

\begin{prop}[{\cite[Chapter 3, Theorem 3]{kushner1967stochastic}, \cite[Section 2.2]{steinhardt2012finite}}]
For a safe $\safeset \subset \mathbb{R}^n$ and a set of initial conditions $\initialset \subset \safeset$, let the function $B: \mathbb{R}^n \mapsto \mathbb{R}_{\geq 0}$ be a \gls{sbf} for System \eqref{eq:system}, and let the positive integer $\timehorizon$ be a time horizon. Then, it holds that $\zeta(\safeset, \timehorizon) \geq 1 - (\gamma + \cmartingale \timehorizon)$.
\label{prop:barrier_prob_safety}
\end{prop}
Thanks to Proposition \ref{prop:barrier_prob_safety}, one can formulate the search for a barrier certificate as the following infinite-dimensional stochastic optimization problem
\begin{equation}
    \begin{aligned} 
        \min_{B \in \mathcal{M}, c, \gamma} & \quad \gamma + c \timehorizon \\
        \mathrm{s.t.} & \quad \eqref{eq:initial_set_constraint},~\eqref{eq:unsafe-set-cond},~\eqref{eq:c-martingale},
    \end{aligned}
    \tag{BP}
    \label{eq:infinite-dim-BP}
\end{equation}
where $\mathcal{M}$ represents the set of non-negative measurable functions in $\mathbb{R}^n$. Two challenges emerge when solving problem \eqref{eq:infinite-dim-BP} to solve Problem \ref{prob:MainProblem}: (i) $\mathcal{M}$ is an infinite-dimensional space, and (ii) constraint \eqref{eq:c-martingale} involves computing an expectation over an unknown probability distribution. To address the first challenge, we restrict the search for barrier functions to the class of \acrlong{pwa} functions given by
\begin{equation}
    \begin{aligned}
        \mathcal{M}' = \{ &B: \mathbb{R}^n \mapsto \mathbb{R}_{\geq 0}: \\ &B(x,\theta) = \max \{ B_1(x,\theta), \ldots, B_{\numberregions}(x,\theta) \} \},  
    \end{aligned}
    \label{eq:parametric-func-class-BP}
\end{equation}
with $\theta \in \mathbb{R}^{\ell(n+1)}$ representing the set of parameters $(u_i,v_i) \in \mathbb{R}^{(n+1)}$ that define the barrier function, i.e.,
\[
    B_i(x,\theta) = \left\{
                \begin{matrix}
                    u_i^\top x + v_i, & \text{ if } x \in \region_i, \\
                    0, & \text{ otherwise},
                \end{matrix}   
            \right.
\]
where the polyhedral sets $\xoverline{\region_i} = \{ x \in \mathbb{R}^n: H_i x \leq h_i \}$, for all $i \in \{1, \ldots, \bar{\numberregions}\}$, constitute a partition of $\mathbb{R}^n$. While various classes of barrier functions have been considered in the literature \cite{PJP:07}, \cite{Santoyo2021}, \cite{steinhardt2012finite}, in this work, we focus on \gls{pwa} barrier functions, as they are expressive enough to approximate any non-linear function arbitrarily well and, as we will show in Section \ref{sec:data_driven_uncertain_stochastic_barriers}, this choice leads to a \gls{lp} program to synthesise a \gls{sbf}, thus guaranteeing efficiency. 
For the second challenge, in the next section we develop a new approach to create a chance-constrained approximation of $\eqref{eq:infinite-dim-BP}$ using a novel constraint tightening technique.

\section{An inner chance-constrained approximation of Problem \eqref{eq:infinite-dim-BP}}
\label{sec:chance-constrained-apprx}

Solving Problem \eqref{eq:infinite-dim-BP} is challenging, even in the case that the probability distribution underpinning \eqref{eq:c-martingale} is known. 
In this section, we show how to relax Problem \eqref{eq:infinite-dim-BP} using a reformulation in terms of a chance-constrained optimisation problem whose feasible set is contained in the set of feasible solutions of \eqref{eq:infinite-dim-BP}.
Such ideas have never been used in this context. Hence, we depart completely from the approaches taken by \cite{Jagtap2020}, \cite{PJP:07}, \cite{SALAMATI20217}, \cite{Santoyo2021}, \cite{steinhardt2012finite}, and \cite{vzikelic2023learning}, which either rely on approximating Constraint \eqref{eq:c-martingale} with the empirical distribution, make strong assumptions about the noise distribution to analytically compute the expectation and recast \eqref{eq:c-martingale} as a convex constraint, or rely on convex over-approximations of the expectation to (conservatively) verify \eqref{eq:c-martingale}. Furthermore, due to the inner approximation of \eqref{eq:infinite-dim-BP} in terms of a chance-constrained problem, our approach opens the road to use the tools of scenario optimisation discussed in Section \ref{subsec:prelim_scenario_theory} to obtain strong sample complexity guarantees on the safety probability of System~\eqref{eq:system}.

To this end, fix any $B \in \mathcal{M}$ and $x\in \mathbb{R}^n$, and let $E \in \mathcal{F}$ be a measurable set. Then, observe that
\begin{equation}
    \begin{aligned}
         \mathbb{E}\{ &B (f(x) + \noise(\omega)) \} = \int_E B (f(x) + \noise(\omega)) \mathbb{P}(d\omega) \\
         &+ \int_{E^c} B (f(x) + \noise(\omega)) \mathbb{P}(d\omega).
    \end{aligned}
    \label{eq:initial-intuition}
\end{equation}
Under the assumption that $B(x)\leq M$ for any $x\in \mathbb{R}^n$, which can always be enforced for \glspl{sbf}, the second term on the right-hand side of \eqref{eq:initial-intuition} can be bounded by $M \mathbb{P}\{ E^c \}$. Our main idea then is to choose a particular $E^c$ that allows us to control the right-hand side of \eqref{eq:initial-intuition}. Such an intuitive reasoning is made formal in the next lemma.

\begin{lemma}
    For $B \in \mathcal{M}$ with $\lVert B \rVert_{\infty} = M$, let $c$ and $\noise$ be as in \eqref{eq:infinite-dim-BP}. Define the set
    \begin{equation}
        \begin{aligned}
            E = \{ \omega \in \Omega: &\,B(f(x) + \noise(\omega)) + \buffervar \leq B(x) + c, \\ 
             &\text{ for all } x \in X_s\},
        \end{aligned}
        \label{eq:set-def-main-lemma}
    \end{equation}
    for some constant $\buffervar \geq 0$.
    If there exists an $\epsilon \in (0,1)$ such that $\buffervar \geq M \frac{\epsilon}{1-\epsilon}$ and $\mathbb{P}\{E^c \} \leq \epsilon$, then we have $\mathbb{E}\{B(f(x) + \noise(\omega))\} \leq B(x) + c$ for all $x \in X_s$. 
    \label{lemma:controlling-expectation}
\end{lemma}

\begin{proof}
    The proof of our result is based on \eqref{eq:initial-intuition}. We first remark that $E$ is a measurable set (see Appendix \ref{subsec:appendix_proof_measurability}). Next, fix any $x \in X_s$ and define the set 
    \begin{equation}
        E_x = \{\omega \in \Omega: B(f(x) + \noise(\omega)) + \buffervar \leq B(x) + c\}. 
        \label{eq:lemma-intermediate-set}
    \end{equation}
    Notice that, by definition, we have $E \subset E_x$. Furthermore, the set $E_x$ is also measurable due to the assumptions on $B$, $f$ and $\eta$, and the fact that measurability is closed under composition. Using \eqref{eq:initial-intuition}, we then obtain
    \begin{equation}
        \begin{aligned}
            \mathbb{E}\{ B(f(x) + \noise(\omega))\} &\leq (B(x) + c - \buffervar) \mathbb{P}\{E_x\} \\
            &+ M\mathbb{P}\{E_x^c\},
        \end{aligned}
        \label{eq:proof-SBF-lemma-1}
    \end{equation}
    where the first term on the right-hand side of \eqref{eq:proof-SBF-lemma-1} follows from the definition of $E_x$ in \eqref{eq:lemma-intermediate-set}, and the second by the uniform boundedness condition on $B$. It remains to show that
    \begin{equation}
        (B(x) + c - \buffervar) \mathbb{P}\{E_x\} + M\mathbb{P}\{E_x^c\} \leq B(x) + c.
    \end{equation}
    To this end, it suffices to show that 
    \begin{equation}
        -\buffervar \mathbb{P}\{E_x\} + M \mathbb{P}\{E_x^c\} \leq 0
        \label{eq:important-equation}
    \end{equation}
    holds, since $B(x) + c$ is non-negative and $P\{E_x\} \leq 1$ by definition. 
    Substituting the two inequalities $\mathbb{P}\{E_x\} \geq 1-\epsilon$ and $\mathbb{P}\{E_x\} \leq \epsilon$ into the left-hand side of \eqref{eq:important-equation} we obtain
    \[
    -\buffervar \mathbb{P}\{E_x\} + M \mathbb{P}\{E_x^c\} \leq -\buffervar(1-\epsilon) + M \epsilon,
    \]
    whose right-hand side is less than or equal to zero due to the fact that  $\buffervar \geq M \frac{\epsilon}{1-\epsilon}$. This concludes the proof of the lemma. 
\end{proof}

Lemma \ref{lemma:controlling-expectation} is enabled by the chance-constraint tightening variable $\buffervar$ through the condition that $\buffervar \geq M \frac{\epsilon}{1-\epsilon}$. 
A reader may question how to choose $\buffervar$ and $\epsilon$. Choosing $\epsilon$ is a trade-off between assigning less probability mass to the uniform upper bound, which is desirable as the uniform upper bound represents a worst-case for barrier value at the next step, is and the amount of data required when applying the scenario approach theory. In our experiments (see Section \ref{sec:experiments}), we employ $\epsilon = 0.005$.
Once $\epsilon$ is chosen, the optimal choice of $\buffervar$ to minimize $\cmartingale$ is $\buffervar = M \frac{\epsilon}{1-\epsilon}$, which follows from the fact that $\buffervar$ is on the left-hand side of the inner inequality of $E$ (and $E_x$) with $c$ on the right-hand side and $M \frac{\epsilon}{1-\epsilon}$ is the smallest allowed value of $\buffervar$.

An immediate consequence of Lemma \ref{lemma:controlling-expectation} is the fact that we can obtain an inner approximation of the optimisation problem \eqref{eq:infinite-dim-BP} in terms of a chance-constrained optimisation problem.

\begin{theorem}
Consider the dynamical system given in System \eqref{eq:system}. Then, we have that for all $\epsilon \in (0,1)$ and positive integers $\timehorizon$ and $M \geq 1$ such that $\buffervar \geq M \frac{\epsilon}{1-\epsilon}$, the feasible set of
	\begin{equation}
		\begin{aligned}
			\min_{B \in \mathcal{M}, c, \gamma} & \quad \gamma + \cmartingale \timehorizon\\
			\textrm{s.t.} & \quad \eqref{eq:initial_set_constraint}, \quad \eqref{eq:unsafe-set-cond}, \quad \gamma \geq 0, \quad \cmartingale \geq 0, \\
                & \quad B(x) \leq M,~\forall x \in \mathbb{R}^n \\ 
                & \quad \mathbb{P}\{\omega \in \Omega: B(f(x)+\noise(\omega)) + \buffervar  \\ 
                & \quad \quad  \leq B(x)  + c,~\forall x \in X_s\} \geq 1-\epsilon,
		\end{aligned}
		\tag{CCBP}
		\label{eq:ccp_opt_pro}
	\end{equation}
	is contained in the feasible set of \eqref{eq:infinite-dim-BP}.	
	\label{theorem:barrier-ccp}
\end{theorem}

\begin{proof}
Constraints \eqref{eq:initial_set_constraint} and \eqref{eq:unsafe-set-cond} are shared between Problem \eqref{eq:infinite-dim-BP} and \eqref{eq:ccp_opt_pro}. Thus, it is sufficient to show that the chance-constraint of Problem \eqref{eq:ccp_opt_pro} implies that constraint \eqref{eq:c-martingale} is satisfied. Rewriting the chance-constraint as $\mathbb{P}\{E\} \geq 1 - \epsilon$ where $E$ is defined as in Lemma \ref{lemma:controlling-expectation}, it holds that $\mathbb{P}\{E^c\} \leq \epsilon$. From this, the assumption $\buffervar \geq M \frac{\epsilon}{1-\epsilon}$, and the constraint that $B$ is uniformly bounded by $M$, the conditions of Lemma \ref{lemma:controlling-expectation} are satisfied. Thus, it holds that if the chance constraint of Problem \eqref{eq:ccp_opt_pro} is satisfied, then constraint \eqref{eq:c-martingale} is satisfied.
\end{proof}

Unfortunately, the computational burden of solving Problem \eqref{eq:ccp_opt_pro} is not negligible, mainly due to the quantification over all $x \in \safeset$ and to the non-linear function $f$. In Section \ref{sec:uncertain_dynamics_safety_and_ccbp}, we will show how the framework developed in Section \ref{subsec:prelim_lbp} can alleviate this burden.

\subsection{Over-approximating general non-linear systems with uncertain PWA dynamics}
\label{sec:uncertain_dynamics_safety_and_ccbp}

In Section \ref{subsec:prelim_lbp}, we showed that any locally Lipschitz function can be over-approximated by uncertain \gls{pwa} functions, and by extension non-linear systems of the form of System \eqref{eq:system} can be over-approximated by uncertain \gls{pwa} systems.
More formally, an uncertain \gls{pwa} over-approximation of System \eqref{eq:system} is described as follows.
\begin{equation}
    \label{eq:uncertain_pwa_system}
    x(k+1) \in F(x(k)) + \noise(k),
\end{equation}
where $F:\mathbb{R}^n \rightrightarrows \mathbb{R}^n$ is a set-valued mapping defined as   
\begin{equation}
    F(x(k)) = \left\{\hat{f}(x(k), \alpha): \alpha \in [0,1]\right\},
    \label{eq:set-valued-mapping-def}
\end{equation}
with $\hat{f}$ being a set-valued \gls{pwa} function in the uncertain parameter $\alpha$ and such that for any $x\in \mathbb{R}^n$ it holds that $f(x)\in F(x)$. 
Such a definition of a uncertain \gls{pwa} over-approximation of System \eqref{eq:system} guarantees that between System \eqref{eq:system} and \gls{pwa} System \eqref{eq:uncertain_pwa_system} there is a \emph{behavioural inclusion} relation, that is, for any $x\in \mathbb{R}^m$ there exists an $\alpha \in [0, 1]$ such that $\hat{f}(x, \alpha) = f(x)$. 

Intuitively, we would like to synthesise a \gls{sbf} for System \eqref{eq:uncertain_pwa_system} and rely on the behavioural relation described above to ensure that the synthesised \gls{sbf} also is an \gls{sbf} for System \eqref{eq:system}. This is what we do in the following Theorem, where we extend Theorem \ref{theorem:barrier-ccp} to an uncertain \gls{pwa} over-approximation of System \eqref{eq:system}.

\begin{theorem}
    Consider the dynamical system given in \eqref{eq:system} and assume access to an uncertain \gls{pwa} over-approximation \eqref{eq:uncertain_pwa_system} of the system. Then, we have that for all $\epsilon \in (0,1)$ and positive integer $\timehorizon$, if there exist $M \geq 1$ and $\buffervar \geq M \frac{\epsilon}{1-\epsilon}$, then the feasible set of
	\begin{equation}
		\begin{aligned}
			\min_{B \in \mathcal{M}, c, \gamma} & \quad \gamma + \cmartingale \timehorizon\\
			\textrm{s.t.} & \quad \eqref{eq:initial_set_constraint}, \quad \eqref{eq:unsafe-set-cond}, \quad \gamma \geq 0, \quad \cmartingale \geq 0, \\
                & \quad B(x) \leq M,~\forall x \in \mathbb{R}^n \\ 
                & \quad \mathbb{P}\{\omega \in \Omega: B(y+\noise(\omega)) + \buffervar  \\ 
                & \quad \quad  \leq B(x)  + c,~\forall x \in X_s,~ \forall y \in F(x)\} \geq 1-\epsilon,
		\end{aligned}
		\tag{UCCBP}
		\label{eq:uncertain_ccp_opt_pro}
	\end{equation}
	is contained in the feasible set of \eqref{eq:infinite-dim-BP}.	
	\label{theorem:uncertain_barrier-ccp}
\end{theorem}

\begin{proof}
    Constraints \eqref{eq:initial_set_constraint} and \eqref{eq:unsafe-set-cond} are imposed directly in \eqref{eq:uncertain_ccp_opt_pro}. Since $f(x) \in F(x)$ implies that there exists an $\alpha \in [0, 1]$ such that $\hat{f}(x, \alpha) = f(x)$, it holds that $B(f(x) + \noise(\uncertaintyelement)) \leq \sup_{y \in F(x)} B(y+\noise(\uncertaintyelement))$. Therefore, the chance constraint of \eqref{eq:uncertain_ccp_opt_pro} implies the chance constraint of \eqref{eq:ccp_opt_pro}. By Theorem \ref{theorem:barrier-ccp} and transitivity of the subset relation, the feasible set of \eqref{eq:uncertain_ccp_opt_pro} is contained in the feasible set of \eqref{eq:infinite-dim-BP}.
\end{proof}

\section{Data-driven approximation of \eqref{eq:uncertain_ccp_opt_pro} for the class of uncertain dynamical systems}
\label{sec:data_driven_uncertain_stochastic_barriers}
To reformulate \eqref{eq:uncertain_ccp_opt_pro} as a robust \gls{lp} problem, we need to introduce some mathematical notation. Let $\regionset = \{\region_1, \ldots, \region_{\numberregions}\}$ be a partition of the state space associated with a barrier function, as described in Section \ref{sec:stochastic_barrier_functions}. Let us also denote four collections of indices by
\begin{equation}
        \begin{aligned}
        &I = \{1,\ldots, \numberregions\},\\
             &I_\initialset = \{i \in I : \region_i \cap \initialset \neq \emptyset \}, \\
        	&I_\safeset = \{i \in I : \region_i \cap \safeset \neq \emptyset \}, \\ 
        	&I_\unsafeset = \{i \in I : \region_i \cap \unsafeset \neq \emptyset \},
        \end{aligned}
    \label{eq:indices-barrier}
\end{equation}
which represent the collection of all indices, and the elements of $\regionset$ with non-empty intersections with the initial state, the safe set, and unsafe set, respectively. 
Finally, for each pair $(i,j) \in I_\safeset \times I$, we denote the set 
\begin{equation}
    \begin{aligned}
    Q_{ij}(\uncertaintyelement) = \left\{ x \in \region_i : \exists y \in F(x), y + \noise(\uncertaintyelement) \in \region_j \right\},
    \label{eq:random-subset-partition}
    \end{aligned}
\end{equation}
representing the subset of $\region_i$ with $i$ belonging to $I_{\safeset}$ and that is mapped to $\region_j$ under a given realization of the noise. A pictorial example of $Q_{ij}(\uncertaintyelement)$ can be found in Figure \ref{fig:preimage_intersection}. 
Leveraging the results of Theorem \ref{theorem:uncertain_barrier-ccp} and Proposition \ref{prop:scenario_approach_theory} and using the notation we have introduced so far, we obtain the following intermediate result. The goal of Lemma \ref{lemma:scenario-barrier-design} below is two fold: (i) to impose piecewise constraints on the barrier synthesis and (ii) apply scenario approach theory to obtain a tractable solution in the face of an unknown noise distribution.

\begin{lemma}
Let $\sampleset = \{ \uncertaintyelement_1, \ldots, \uncertaintyelement_N \}$ be a collection of $N$ independent samples from the noise distribution $\probmeas$. Fix $\epsilon \in (0,1)$, $M \geq 1$, and $\buffervar \geq M\frac{\epsilon}{1-\epsilon}$, and consider the scenario optimisation program 
	\begin{equation}
		\begin{aligned}
			\min_{z} & \quad \gamma + \cmartingale \timehorizon \\
			\textrm{s.t.} & \quad \gamma \geq 0, \quad \cmartingale \geq 0, \\[1ex]
            &\quad \begin{aligned}
			    &\begin{aligned}
    				&  B_i(x, \theta) \in [0, M], &&\forall x \in \region_i, \, i \in I, \\[1ex]
    				&  B_i(x, \theta) \leq \gamma, &&\forall x \in \region_i, \, i \in I_\initialset, \\[1ex]
    				&  B_i(x, \theta) \geq 1, &&\forall x \in \region_i, \, i \in I_\unsafeset,
			    \end{aligned}\\
                &\begin{aligned}
				    B_j(y + \noise(\uncertaintyelement),\theta) + \buffervar \leq B_i(x,\theta) + c,& \\[1ex]
                    \forall (\uncertaintyelement, i, j) \in \sampleset \times I_\safeset \times I, \, y \in F(x), \, x \in Q_{ij}(\uncertaintyelement)&
                \end{aligned}
			\end{aligned}
		\end{aligned}
		\tag{SBP}
		\label{eq:barrier_sp_opt_pro}
	\end{equation}
where $d = 2 + \numberregions (n + 1)$ is the number of decision variables in Problem \eqref{eq:uncertain_ccp_opt_pro}, and $z = (\gamma, c, \theta)$ is the collection of optimisation variables. Then, with probability at least $1 - \beta$, where $\beta$ is the right-hand side of the inequality in Proposition \eqref{prop:scenario_approach_theory}, the optimal solution of \eqref{eq:barrier_sp_opt_pro} satisfies the constraints of Definition \ref{defi:Barrier_certificate}.
	\label{lemma:scenario-barrier-design}
\end{lemma}
\begin{figure}
    \centering
    \includegraphics[width=0.45\textwidth]{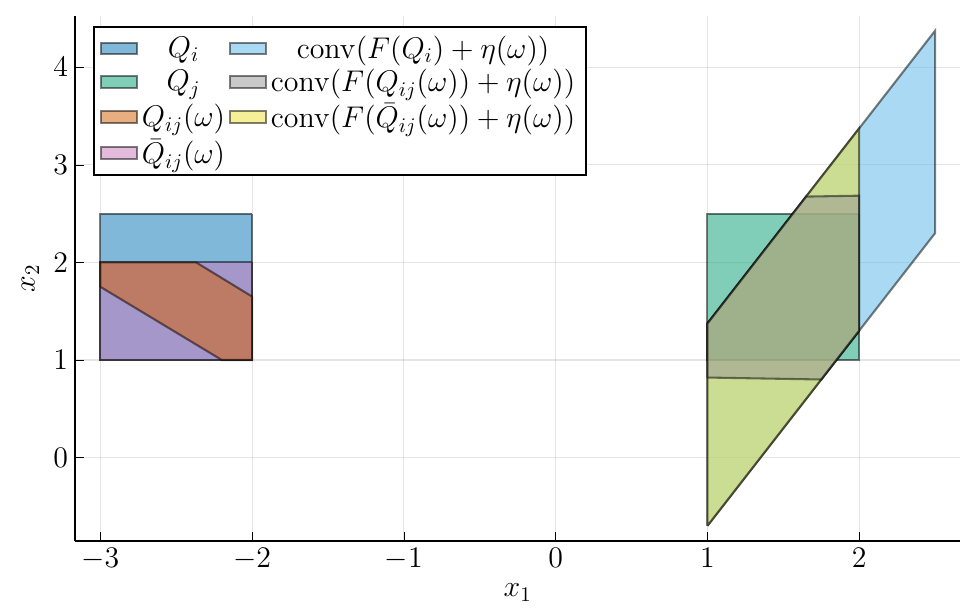}
    \caption{Given two regions $\region_i, \region_j$ and a realisation of the noise $\uncertaintyelement$, the set $Q_{ij}(\uncertaintyelement)$ represents the subset of $x\in \region_i$ such that $\hat{f}(x, \alpha) + \noise(\uncertaintyelement) \in \region_j$ for some $\alpha \in [0, 1]$. In other words, $Q_{ij}(\uncertaintyelement)$ is the subset of $\region_i$ that can reach $\region_j$ given the realisation of the noise $\uncertaintyelement$. Unfortunately, $Q_{ij}(\omega)$ is not easily computed (the example above is an exception, see Fig. \ref{fig:non-convex_preimage}) and thus in Section \ref{subsec:over-approximate_preimage}, we will compute an over-approximation $\xoverline{Q}_{ij}(\uncertaintyelement)$.}
    \label{fig:preimage_intersection}
\end{figure}

\begin{remark}
The state-of-the-art literature on \acrfull{saa} for \gls{sbf} design relies on Chebyshev’s inequality \cite{SALAMATI20217} to bound the probability of satisfaction, which yields a sample complexity proportional to $O(1/\beta)$. Instead, under the assumption that the barrier is uniformly upper bounded, which is always the case for our method, in Lemma \ref{lemma:scenario-barrier-design} we can rely on Hoeffding's inequality, which yields a sample complexity proportional to $O(\log(1/\beta))$. \qed
\end{remark}

The relevance of Lemma \ref{lemma:scenario-barrier-design} towards the general framework proposed in the paper can be summarised by two main points: (1) Lemma \ref{lemma:scenario-barrier-design} establishes a sufficient condition to enforce the constraints of \eqref{eq:uncertain_ccp_opt_pro} by restricting the attention to each element of the partition $\regionset$ individually; (2) it allows us to use the duality results of Section \ref{subsec:prelim_LP} to obtain a computationally tractable reformulation of the optimisation problem in Lemma \ref{lemma:scenario-barrier-design}
into a robust \gls{lp}.  
To illustrate the latter point, let $i \in I_0$ and consider its corresponding partition $\region_i = \{x: H_i x \leq h_i\}$. Observe that $B(x,\theta) \leq \gamma$ for all $x \in \region_i$ can be written as $h_i^\top \lambda \leq \barrierc_i$ and $H_i^\top \lambda = - \barrierl_i$ hold where $\lambda_i$ is a non-negative decision variable. This follows from the fact that within the region $\region_i$ the function $B(x, \theta) = B_i(x, \theta)$ is affine and via strong duality (Prop. \ref{prop:main_result_robust_LP}), we can rewrite this robust constraint to two regular linear constraints.
The rewrite of the non-negativity, uniform upper bound, and unsafe set robust constraints (see Def. \ref{defi:Barrier_certificate}) to regular linear constrains follow a similar argument. Hence, for brevity, we omit the reformulation of these.

The remaining constraint $B(y + \noise(\uncertaintyelement), \theta) + \buffervar \leq B(x, \theta) + \cmartingale$ for all $\uncertaintyelement \in \sampleset$, $x \in \safeset$, $y \in F(x)$, requires more care, yet Lemma \ref{lemma:scenario-barrier-design} also enables a computationally tractable reformulation of this.
To this end, consider $(i, j) \in I_\safeset \times I$ and $\uncertaintyelement \in \sampleset$.
A challenge in reformulating the constraint as a linear constraint is that the set $Q_{ij}(\uncertaintyelement)$ is not a polyhedron or even convex (see Fig. \ref{fig:non-convex_preimage}). For now, we assume access to a polyhedral over-approximation $Q_{ij}(\uncertaintyelement) \subset \xoverline{Q}_{ij}(\uncertaintyelement) = \{x : H_{ij\uncertaintyelement}x \leq h_{ij\uncertaintyelement} \}$.
Then if we impose the constraint for all $x \in \xoverline{Q}_{ij}(\uncertaintyelement)$, then it trivially follows that it also holds for all $x \in Q_{ij}(\uncertaintyelement)$.
We will defer the discussion of how to compute $\xoverline{Q}_{ij}(\uncertaintyelement)$ to Section \ref{subsec:over-approximate_preimage}.

\begin{prop}\label{prop:finite_cmartingale_constraints}
    Let $(i, j) \in I_\safeset \times I$ and $\uncertaintyelement \in \sampleset$ be given. Assume that there exists two affine functions such that $\underline{A}_ix + \underline{b}_i \leq F(x) \leq \xoverline{A}_ix + \xoverline{b}_i$ for all $x \in \region_i$. Then it holds that $B_j(y + \noise(\uncertaintyelement), \theta) + \buffervar \leq B_i(x, \theta) + \cmartingale$ for all $x \in \xoverline{Q}_{ij}(\uncertaintyelement), y \in F(x)$ if and only if the following constraints hold
    \begin{align*}
        h_{ij\uncertaintyelement}^\top \underline{\lambda}_{ij\uncertaintyelement} &\leq \barrierc_i - \barrierc_j - \barrierl_j^\top (\underline{b}_i + \noise(\uncertaintyelement)) + c - \buffervar, \\
        H^\top_{ij\uncertaintyelement} \underline{\lambda}_{ij\uncertaintyelement} &= \underline{A}_i^{\top} \barrierl_j - \barrierl_i,\\
        h_{ij\uncertaintyelement}^\top \xoverline{\lambda}_{ij\uncertaintyelement} &\leq \barrierc_i - \barrierc_j - \barrierl_j^\top (\xoverline{b}_i + \noise(\uncertaintyelement)) + c - \buffervar, \\
        H^\top_{ij\uncertaintyelement} \xoverline{\lambda}_{ij\uncertaintyelement} &= \xoverline{A}_i^{\top} \barrierl_j - \barrierl_i,
    \end{align*}
    where $\underline{\dualvariable}_{ij\uncertaintyelement}$, $\xoverline{\dualvariable}_{ij\uncertaintyelement}$ are a non-negative dual variables.
\end{prop}

Collecting together all finite sets of constraints, a finite representation of the semi-infinite program \eqref{eq:barrier_sp_opt_pro} is given as
\begin{equation}
\begin{aligned}
    \min_{z} & \quad \gamma + \cmartingale \timehorizon\\
    \subjectto & \quad \gamma \geq 0, \, \cmartingale \geq 0,  \\[1ex]
        & \quad \begin{alignedat}{3}
        & \text{(Non-negativity)} \\
        & h_i^\top \underline{\nu}_{i} \leq \barrierc_i, \;  H_i^\top \underline{\nu}_i = -\barrierl_i, \text{ for all } i \in I, \\[1ex]
        & \text{(Uniform upper bound)} \\
        & h_{i}^\top \xoverline{\nu}_i \leq M - \barrierc_i, \; H_{i}^\top \xoverline{\nu}_i = \barrierl_i, \text{ for all } i \in I, \\[1ex]
        & \text{(Initial set)} \\
        & h_{i0}^\top \mu^0_i\leq \gamma - \barrierc_i, \;H^\top_{i0} \mu^0_i = \barrierl_i, \text{ for all } i \in I_\initialset,  \\[1ex]
        & \text{(Unsafe set)} \\
        & h_{i}^\top \mu_i^u\leq \barrierc_i - 1, \; H^\top_{i} \mu_i^u = -\barrierl_{i}, \text{ for all } i \in I_\unsafeset, \\[1ex]
        & \text{(One-step constraints)} \\
        & h_{ij\uncertaintyelement}^\top \underline{\lambda}_{ij\uncertaintyelement} \leq \barrierc_i - \barrierc_j - \barrierl_j^\top (\underline{b}_i + \noise(\uncertaintyelement)) + c - \buffervar, \\[1ex]
        & H^\top_{ij\uncertaintyelement} \underline{\lambda}_{ij\uncertaintyelement} = \underline{A}_i^{\top} \barrierl_j - \barrierl_i, \\[1ex]
        & h_{ij\uncertaintyelement}^\top \xoverline{\lambda}_{ij\uncertaintyelement} \leq \barrierc_i - \barrierc_j - \barrierl_j^\top (\xoverline{b}_i + \noise(\uncertaintyelement)) + c - \buffervar, \\[1ex]
        & H^\top_{ij\uncertaintyelement} \xoverline{\lambda}_{ij\uncertaintyelement} = \xoverline{A}_i^{\top} \barrierl_j - \barrierl_i, \text{ for all } \uncertaintyelement \in \sampleset, \\
        & \qquad \text{ for all } (i, j) \in I_\safeset \times I,
    \end{alignedat}
\end{aligned}
\tag{FSBP}
\label{eq:barrier_sp_opt_pro_finite}
\end{equation}
where $\underline{\nu}_{i}, \xoverline{\nu}_i, \mu^0_i, \mu_i^u, \underline{\dualvariable}_{ij\uncertaintyelement}, \xoverline{\dualvariable}_{ij\uncertaintyelement}$ are non-negative dual variables.
$(H_{i0}, h_{i0})$ denotes the half-space representation of $\region_i \cap \initialset$.

In Corollary \ref{corollary:scenario-barrier-design}, we put together the results we introduced and show how probabilistic safety can be computed via the scenario approach using samples of the random variable $\noise(\uncertaintyelement)$.

\begin{corollary}\label{corollary:scenario-barrier-design}
    Assume that $\sampleset = \{ \uncertaintyelement_1, \ldots, \uncertaintyelement_N \}$ is a collection of $N$ independent samples from the noise distribution $\probmeas$. Fix $\epsilon \in (0,1)$, $M \geq 1$, and $\buffervar \geq M\frac{\epsilon}{1-\epsilon}$, and let $\beta$ be defined as Prop. \ref{prop:scenario_approach_theory} where $d = 2 + \numberregions (n + 1)$ is the number of (non-dual) decision variables in Problem \eqref{eq:barrier_sp_opt_pro_finite}.
    Consider the optimal (primal) solution $z^\star(\sampleset)=(\cmartingale^\star,\gamma^\star,\theta^\star)$ of Problem \eqref{eq:barrier_sp_opt_pro_finite}.
    Then, with confidence $1-\beta$, it holds that \[\zeta(\safeset, \timehorizon) \geq 1 - (\gamma^\star + \cmartingale^\star \timehorizon).\]
\end{corollary}

Thus, by solving Problem \eqref{eq:barrier_sp_opt_pro_finite}, we can certify probabilistic safety with high confidence. Note that na\"ively trying to solve \eqref{eq:barrier_sp_opt_pro_finite} can soon become intractable on contemporary hardware, due to both memory requirements and computational time. In particular, the cardinality of the Cartesian product of $I_{\safeset}$, $I$, and $\sampleset$ can already be prohibitively large for relatively small systems. In the next section, we will discuss algorithmic strategies that make Problem \eqref{eq:barrier_sp_opt_pro_finite} computationally tractable.

\section{Algorithms for program construction}
\label{sec:computational_methods_and_efficiency}
In this section, we discuss three aspects that allows one to solve Problem \eqref{eq:barrier_sp_opt_pro} efficiently. Namely, in Section \ref{subsec:over-approximate_preimage}, we discuss how to compute a polyhedral over-approximation of $Q_{ij}(\uncertaintyelement)$. In Subsection \ref{subsec:pwa_cert_convex_hull_noise}, we introduce an a-priori sample discarding procedure for Problem \eqref{eq:barrier_sp_opt_pro_finite}, which guarantees the same optimal solution, but allowing one to consider less samples in the optimisation problem. Finally, in Subsection \ref{subsec:spatial_indexing_pwa}, we employ spatial indexing methods to efficiently find the set of triplets with a non-empty $Q_{ij}(\uncertaintyelement)$. More specifically, we will rely on the fact that, often, for many triplets $(i, j, \uncertaintyelement)$, the set $Q_{ij}(\uncertaintyelement)$ is empty, thus the martingale constraint is trivially satisfied. That is, $Q_{ij}(\uncertaintyelement)$ is empty if $\im_i(\region_{i}, \uncertaintyelement) \cap \region_{j} = \emptyset $ where the image for region $\region_i$ is defined as 
\begin{equation}
    \im(\region_i, \uncertaintyelement) = \left\{ y + \noise(\uncertaintyelement) : x \in \region_i, y \in F(x) \right\}.
\end{equation}

\begin{figure}
    \centering
    \includegraphics[width=0.45\textwidth]{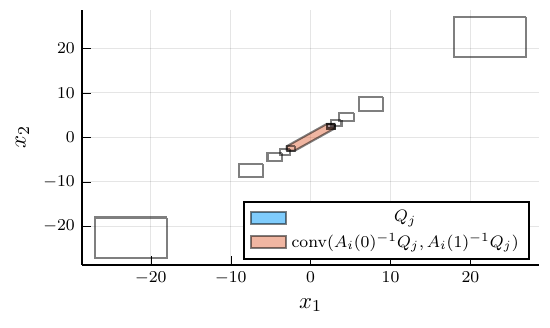}
    \caption{A pictorial example that $Q_{ij}(\uncertaintyelement)$ is not necessarily convex. In this example, we have $\underline{A_i} = -I$, $\xoverline{A_i} = I$, and $\underline{b_i} = \xoverline{b_i} = 0$, which is a valid uncertain affine relaxation of the trivial function $f(x) = 0$ in the non-negative orthant. We consider the region $\region_j = [2, 3]^2$ and plot $A_i(\alpha)^{-1}\region_j$ for 10 different values of $\alpha \in [0, 1]$. Note that in this case $\xoverline{A_i} = I$, hence $\region_j = A_i(1)^{-1}\region_j$ and $\region_j$ (blue) is contained in the convex hull (pink).
    $Q_{ij}(\uncertaintyelement)$ can take on complex shapes and no method exists for exactly computing $Q_{ij}(\uncertaintyelement)$. Therefore, we compute a sound over-approximation $\xoverline{Q}_{ij}(\uncertaintyelement) \subset Q_{ij}(\uncertaintyelement)$.}
    \label{fig:non-convex_preimage}
\end{figure}

\subsection{Over-approximation of $Q_{ij}(\uncertaintyelement)$}\label{subsec:over-approximate_preimage}
As illustrated in Fig. \ref{fig:preimage_intersection}, computing a polyhedral over-approximation of $Q_{ij}(\uncertaintyelement)$ is challenging, as $Q_{ij}(\uncertaintyelement)$ is (possibly) non-convex due to the uncertain\footnote{If $F$ is a deterministic affine transformation and $\region_j$ is a polyhedron, then $Q_{ij}(\uncertaintyelement)$ is also a polyhedron and analytical methods for computation exist.} affine transformation $F$.
Furthermore, it is not sufficient to compute the convex hull for the vertices of the uncertainty variable $\alpha$, that is,
\begin{equation}
    \begin{aligned}
	\conv & \left( \left\{ x \in \region_i : \hat{f}_i(x, 0) + \noise(\uncertaintyelement)) \in \region_j \right\} \cup \right. \\
        & \;\, \left. \left\{ x \in \region_i : \hat{f}_i(x, 1) + \noise(\uncertaintyelement)) \in \region_j \right\}\right),
    \end{aligned}
\end{equation}
as shown in Fig. \ref{fig:preimage_intersection}.
However, we note that by definition $Q_{ij}(\uncertaintyelement) \subseteq \region_i$, that is, $\region_i$ is a, generally conservative, polyhedral over-approximation of $Q_{ij}(\uncertaintyelement)$. Hence, our approach to find a polyhedral over-approximation of $Q_{ij}(\uncertaintyelement)$, denoted $\xoverline{Q}_{ij}(\uncertaintyelement)$, is to start from $\region_i$ and then iteratively removing subsets of $\region_i \setminus Q_{ij}(\uncertaintyelement)$.
To accomplish this, we rely on repeated bisection.
To simplify the presentation, in what follows, we assume that $\region_i$ is a hyperrectangle. Note, however, that the procedure generalises to compact polyhedra in half-space representation.

Fig. \ref{fig:example_bisection_algorithm} shows an example of the bisection algorithm for two regions $\region_i, \region_j$ and a given sample $\uncertaintyelement$. The bisection is repeated twice along each axis, namely once to increase the lower bound, once to decrease the upper bound. Note that $Q_{ij}(\uncertaintyelement)$, although depicted in Fig. \ref{fig:example_bisection_algorithm}, is unknown and possibly non-convex, and further we cannot readily check if $Q_{ij}(\uncertaintyelement) \subset \xoverline{Q}_{ij}(\uncertaintyelement)$.
However, recall that $Q_{ij}(\uncertaintyelement)$ is the subset of region $\region_{i}$ that under a realisation of the noise $\uncertaintyelement$ reaches region $\region_j$ in one time step.
As a result, we can instead check if $\conv(\xoverline{Q}_{ij}(\uncertaintyelement) + \noise(\uncertaintyelement))$ intersects with $\region_j$.
By applying this algorithm, we compute a small over-approximation $\xoverline{Q}_{ij}(\uncertaintyelement)$ of $Q_{ij}(\uncertaintyelement)$. Indeed as $\xoverline{Q}_{ij}(\uncertaintyelement)$ is an over-approximation, using $\xoverline{Q}_{ij}(\uncertaintyelement)$ for the constraints in Proposition \ref{prop:finite_cmartingale_constraints} yields a sound, although slightly conservative, solution.

\begin{figure*}
    \centering
    \begin{subfigure}{0.19\textwidth}
        \centering
        \includegraphics[width=\linewidth, page=1]{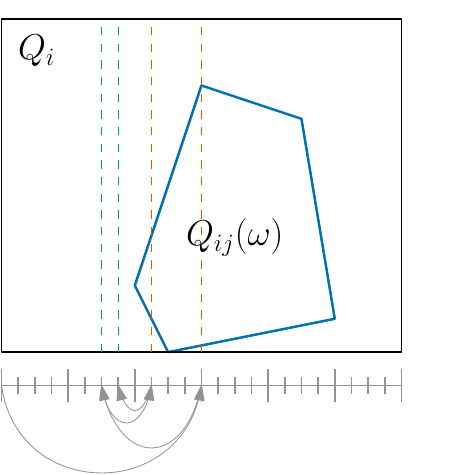}
        \caption{}
    \end{subfigure}
    \begin{subfigure}{0.19\textwidth}
        \centering
        \includegraphics[width=\linewidth, page=2]{figures/bisection.pdf}
        \caption{}
    \end{subfigure}
    \begin{subfigure}{0.19\textwidth}
        \centering
        \includegraphics[width=\linewidth, page=3]{figures/bisection.pdf}
        \caption{}
    \end{subfigure}
    \begin{subfigure}{0.19\textwidth}
        \centering
        \includegraphics[width=\linewidth, page=4]{figures/bisection.pdf}
        \caption{}
    \end{subfigure}
    \begin{subfigure}{0.19\textwidth}
        \centering
        \includegraphics[width=\linewidth, page=5]{figures/bisection.pdf}
        \caption{}
    \end{subfigure}
    \caption{An example of the bisection algorithm to compute the over-approximation $\xoverline{Q}_{ij}(\uncertaintyelement)$ of $Q_{ij}(\uncertaintyelement)$. The set $Q_{ij}(\uncertaintyelement)$ is unknown and possibly non-polyhedral, but $\region_i$ is a sound over-approximation. By bisection from either side (first the lower then the upper bound) along each axis we obtain a smaller over-approximation. We start by bisecting for $x_1$ (\textbf{(a)} and \textbf{(b)}) followed by $x_2$ (\textbf{(c)} and \textbf{(d)}). This results in the over-approximation $\xoverline{Q}_{ij}(\uncertaintyelement)$ in \textbf{(e)}.}
    \label{fig:example_bisection_algorithm}
\end{figure*}

\subsection{Convex hull over the sample set}
\label{subsec:pwa_cert_convex_hull_noise}
To reduce the number of constraints of Problem \eqref{eq:barrier_sp_opt_pro_finite}, observe that the constraints of the problem are affine in $\noise(\uncertaintyelement)$. This implies that the active constraints, also known as support constraints, will always belong to the vertices of the convex hull over $\noise(\sampleset):=\{\noise(\omega_1), \ldots, \noise(\omega_N)\}$, enabling a great reduction in the number of constraints in the program.
That is, we can impose the constraint only on the following set of samples:
\begin{equation}\label{eq:convex_hull_noise}
    \xoverline{\sampleset} = \{ \omega \in \sampleset : \noise(\uncertaintyelement) \in \vertices(\conv(\noise(\sampleset))) \}.
\end{equation}

In the experiments conducted (see Section \ref{sec:experiments}), we find that generally, in practice, the cardinality of $\xoverline{\sampleset}$ is orders of magnitude lower than the cardinality of $\sampleset$. Thus, this can greatly improve the efficiency of our approach.

\begin{remark}
The method presented in this subsection was discovered independently, but is similar to the method presented in \cite{Sartipizadeh2018} with the exception of  that our method requires an exact convex hull rather than an approximate convex hull. The proposed method for sample reduction works for any scenario program that is affine in the random variable $\noise(\uncertaintyelement)$. \qed 
\end{remark}

\subsection{Spatial indexing for intersection search}\label{subsec:spatial_indexing_pwa}
Constructing Problem \eqref{eq:barrier_sp_opt_pro_finite} efficiently is also a non-trivial problem due to memory limits.
In fact, a na\"{i}ve approach to construct the problem is to iterate over all triplets $(i, j, \uncertaintyelement)$ in $I_\safeset \times I \times \sampleset$, check if $Q_{ij}(\uncertaintyelement) \neq \emptyset$, and add a set of constraints if the test is positive. This approach is only tractable for smaller problems as it has  time complexity $\mathcal{O}(\lvert \regionset \rvert^3)$. 
To reduce the complexity, we can exploit methods from database theory; namely spatial indexing, which is the structuring and querying of spatially distributed data, such as maps, with improved computational complexity \cite{10.1145/971697.602266}. 

To apply spatial indexing to our setting, we must first establish the data and query. It holds that $Q_{ij}(\uncertaintyelement) \neq \emptyset$ only if $\im_i(\region_i, \uncertaintyelement) \cap \region_j \neq \emptyset.$
Hence, if we search for regions $\region_j$ that intersect with the image $\im_i(\region_i, \uncertaintyelement)$, we find all pairs $(i, j)$ such that $Q_{ij}(\uncertaintyelement) \neq \emptyset$. For spatial indexing, we focus on R-trees as it is a well-studied and widely available method \cite{10.1145/971697.602266}. The idea is to structure the data in a tree structure and at each node store a \gls{mbr} for the nodes below. Then, querying the tree for the intersection with another region proceeds recursively down the tree, where it is only necessary to search down a branch if the \gls{mbr} of the branch and the query intersect, which is an inexpensive operation by the separating hyperplane theorem \cite{BV:04}. 
Figure \ref{fig:r_tree_spatial_indexing_state_space} shows an example of an R-tree for a partitioned state space.
This method improves complexity by efficiently searching for relevant triplets $(i, j, \uncertaintyelement)$.

\begin{figure}
    \centering
    \includegraphics[width=0.8\linewidth]{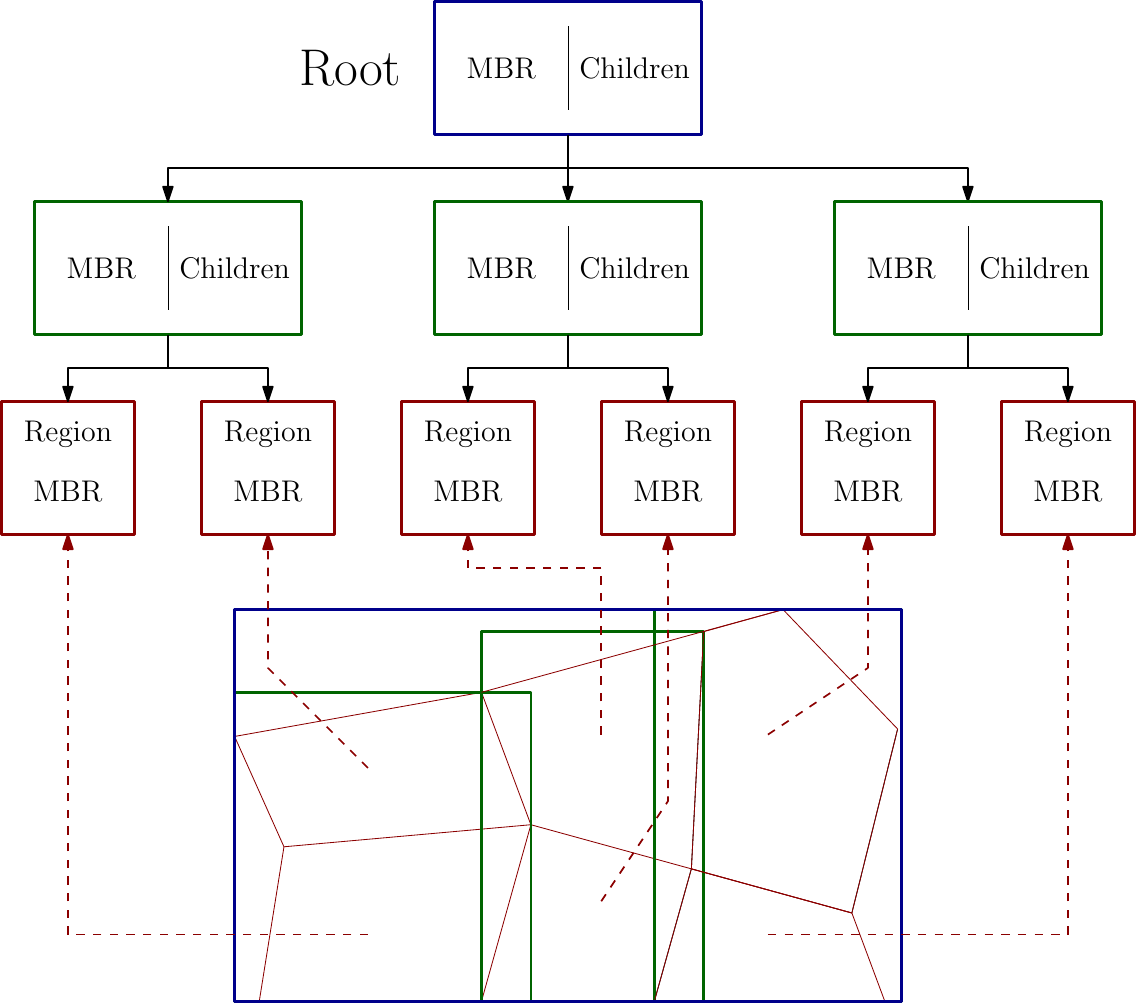}
    \caption{An example of an R-tree applied to a partitioned state space to allow efficient search for regions intersecting with $\im_i(\region_i, \uncertaintyelement)$. The rectangles are the \acrlong{mbr} for each node in the tree.}
    \label{fig:r_tree_spatial_indexing_state_space}
\end{figure}

In summary, to use the framework to compute data-driven safety certificates for non-linear systems: let the nominal dynamics $f$, initial and safe set $\initialset$, $\safeset$, a horizon $T$, and a dataset of samples $\sampleset$ be given. Then start by abstracting the non-linear dynamics $f$ to uncertain \gls{pwa} dynamics $\hat{f}$ using \gls{lbp} techniques. Compute the vertices $\xoverline{\sampleset}$ of the convex hull of $\sampleset$ and discard all interior points. For each region $\region_i$, find, using an R-tree, regions that intersect with the image of the dynamics $\im_i(\region_i, \uncertaintyelement)$ and add constraints accordingly. Solve the \gls{lp} problem \eqref{eq:barrier_sp_opt_pro_finite}, then the solution $z^\star(\xoverline{\sampleset}) = (\cmartingale^\star, \gamma^\star, \theta^\star)$ is a safety certificate $\zeta(\safeset, T) \geq 1 - (\gamma^\star + \cmartingale^\star T) $ with confidence $1 - \beta$ where $\beta$ is defined as in Prop. \ref{prop:scenario_approach_theory}.

\section{Experiments}
\label{sec:experiments}

To support the theoretical results and investigate the efficacy of our approach, we implemented our framework in Julia\footnote{Code is available at \url{https://github.com/DAI-Lab-HERALD/scenario-barrier} under GNU GPLv3 license.} and performed an empirical analysis on various benchmarks. The experiments have been conducted on a computer with an Intel i7-6700k CPU, Nvidia GTX1060 6GB GPU, and 16GB RAM, running Linux 5.10.211-1-MANJARO.
We start by describing the benchmarks followed by the results.
For a comparision with state-of-the-art, we consider the \acrfull{saa} method proposed in \cite{SALAMATI20217}, \cite{Salamati2023}. As this \gls{sbf} synthesis method has been developed specifically for linear or polynomial systems, namely through \gls{sos} optimisation, in order to provide a general baseline, in the case of non-polynomial and/or uncertain systems, we combine it with the method proposed in \cite{mazouzsafety2022}, to find a valid \gls{sbf} in the general case.

\subsection{Benchmarks}
The simplest system considered is an uncertain 1D linear system $x(k + 1) = x(k) + b(\alpha) + \noise$ where $b(\alpha) = -0.05 + 0.1\alpha$, i.e. the uncertainty is in $b(\alpha)$ with uncertainty variable $\alpha \in [0, 1]$. Starting in a set around the origin $\initialset := \{ \lvert x \vert \leq 0.5 \}$, the goal is  to stay within a larger set $\safeset := \{ \lvert x \vert \leq 2.5 \}$ for a horizon $T = 10$. The distribution of the noise is a zero-mean normal distribution with standard deviation $0.01$.

We also consider a 2D system from \cite{Badings2022} representing the longitudinal dynamics of a drone. The coordinates $x_1, x_2$ are the position and velocity, respectively, and the system has the following dynamics
\[
    \setlength\arraycolsep{6pt}
    x(k + 1) = \begin{pmatrix}
        1 & \tau \\ 0 & 1 - \frac{0.1\tau}{m} 
    \end{pmatrix}x(k) + \eta, 
\]
where $m \in [0.75, 1.25]$ and $\eta$ has a zero-mean normal distribution with diagonal covariance of $[0.01 \; 0.01]$. The variable $\tau$ represents the discretisation step, which is set to $\tau = 1.0$. As with the 1D linear system, we certify safety for a horizon $T = 10$.

The third benchmark represents a model of a vehicle travelling down a straight road when it experiences an (uncertain) gust of wind.
The coordinates $x_1, x_2$ represent, respectively, the longitudinal and lateral position of the vehicle. Similar to the drone system, $\tau$ represents the discretisation step. The goal is certify the probability of staying on the road $\safeset := \{ \lvert x_1 \vert \leq 2.5 \}$ for a horizon $T = 10$, when the system evolves according to the following dynamics
\[
    \setlength\arraycolsep{6pt}
    x(k + 1) = \begin{pmatrix}
        1 & 0 \\ 0 & 0.95
    \end{pmatrix}x(k) + \begin{pmatrix}
        \frac{50}{3.6} \cdot \tau\\
        \frac{1}{2} a_{lat} \cdot \tau^2
    \end{pmatrix} + \eta, 
\]
where $\tau = 1$, and $a_{lat} = 0$ for regions where $x_1 \leq 80$ or $x_1 \geq 120$, and $a_{lat} \in [0.0913, 0.364]$ for regions where $x_1 \leq 80$ or $x_1 \geq 120$. $\eta$ has a zero-mean normal distribution with diagonal covariance of $[0.01 \; 0.01]$.

While the previous models were linear, we also considered non-linear models. In particular, we consider \acrfullpl{nndm} with 1 and 2 hidden layers of 64 neurons each modelling a pendulum taken from \cite{mazouzsafety2022}.
Finally, the last benchmark is the 3D model Dubin's car from \cite{Mathiesen2013} for a time horizon $T = 10$. Dubin's car is a unicycle model where the state is $(x, y, \phi)$ with $\phi$ being the heading of the vehicle. We consider a grid-based partitioning of 10 segments along each axis, i.e. 1000 regions. The noise is only applied to the last dimension and has a normal distribution with mean of $60 \cdot \frac{\pi}{180} \approx 1.053$ and standard deviation $0.1$.

\subsection{Results}
Table \ref{tab:results_table} shows a list of results across all benchmarks. Both the safety probability and the computation time are reported as the mean over the 100 trials and for all cases the number of samples is selected to ensure a confidence $1 - \beta = 1 - 10^{-9}$.  From Table \ref{tab:results_table} we observe that, depending on the system, the method can certify safety to $> 99\%$ with high confidence, e.g. $99.5\%$ certified safety for the \gls{nndm} model of a pendulum with 2 layers and 64 neurons.
This certification requires relatively few regions of 10-30 segments per axis. 
Comparing the \gls{nndm} pendulum model with 2 and 3 layers (1 and 2 hidden layers, respectively), the complexity of the nominal dynamics impacts both computation time and certifiable safety, e.g. $99.5\%$ safety probability in $45.0 \mathrm{s}$ vs $97.6\%$ safety probability in $78.5 \mathrm{s}$ for $480$ regions, 2 and 3 layers respectively. This behavior can be explained by \gls{lbp} computing wider uncertain affine transformations lead to more non-empty $Q_{ij}(\omega)$.

Remarkably, the linearity of the underlying system has little impact on the certifiable safety.
This is observed in that both the 1D linear and drone systems exhibit uncertain linear behavior, yet the 1D linear system is certifiable to $50.8\%$ safety while the drone is certifiable to $99.5\%$.
Furthermore, the largest system considered, Dubin's car, which includes trigonometric functions, safety is certified to $99.9\%$.

\begin{table}
    \aboverulesep=0ex
    \belowrulesep=0ex
    \renewcommand{\arraystretch}{1.3}
    \centering
    \caption{Certified safety and computation time using the method explained in Sections \ref{sec:chance-constrained-apprx}-\ref{sec:computational_methods_and_efficiency}. Results are reported as the mean over 100 iterations for each case study. $n$ is the dimensionality of the system and $\numberregions$ is the number of regions. $\zeta(\safeset, \timehorizon)$ is the certified level of safety for $\beta = 10^{-9}$ where $1 - \beta$ is the level of confidence.}
    \label{tab:results_table}
    \vspace{0.2em}
    {\small \begin{tabularx}{\linewidth}{X|c|c|cc}
    \toprule
       System  & $n$ & $\numberregions$ & $\zeta(\safeset, \timehorizon)$ & Time (s)\\\midrule\midrule
         Linear & 1 & 27 & $0.508 $ & $0.239$\\\midrule
         Drone &  2 & 37 & $0.995$ & $41.8$ \\\midrule
         \multirow{4}{*}{Vehicle} & \multirow{4}{*}{2} & 18 & $0.607$ & $0.716$ \\
         & & 42 & $0.709$ & $1.86$ \\
         & & 54 & $0.827$ & $2.20$\\
         & & 150 & $0.995$ & $6.84$\\\midrule
         Pendulum (\gls{nndm}) & \multirow{3}{*}{2} & 120 & $0.344$ & $7.26$ \\
         - 2 layers &              & 240 & $0.756$ & $17.6$ \\
         - 64 neurons     &              & 480 & $0.995$ & $45.0$ \\\midrule
         Pendulum (\gls{nndm}) & \multirow{3}{*}{2} & 120 & $0.254$ & $14.3$ \\
         - 3 layers &             & 240 & $0.374$ & $36.3$ \\
         - 64 neurons      &             & 480 & $0.976$ & $78.5$ \\\midrule
        Dubin's car & 3 & 1000 & $0.999$ & $383$ \\
        \bottomrule
    \end{tabularx}}
\end{table}


\begin{table}
    \aboverulesep=0ex
    \belowrulesep=0ex
    \renewcommand{\arraystretch}{1.3}
    \centering
    \caption{Comparison of our method against \acrfull{saa} combined with the method presented in \cite{mazouzsafety2022}, to synthesise \glspl{sbf} in a data-driven fashion for non-polynomial and uncertain systems. Results are reported as the mean over 100 iterations for each case study. $ 1 - \beta$ is the confidence in the certificate and $\zeta(\safeset, \timehorizon)$ is the certified level of safety. OOM means the certification procedure crashed with an out-of-memory error.}
    \label{tab:results_table_saa}
    \vspace{0.2em}
    {\small \begin{tabularx}{\linewidth}{X|l|ccc}
    \toprule
       System  & Method & $\beta$ & $\zeta(\safeset, \timehorizon)$ & Time (s)\\\midrule\midrule
         \multirow{6}{*}{Linear} &      & $10^{-2}$ & $0.514$ & $0.219$\\
                                 & Ours & $10^{-3}$ & $0.513$ & $0.226$\\
                                 &      & $10^{-4}$ & $0.512$ & $0.234$\\\cmidrule{2-5}
                                 &      & $10^{-2}$ & $0.506$ & $1.91$\\
                                 & SAA  & $10^{-3}$ & $0.506$ & $24.9$\\
                                 & & $10^{-4}$ & - & OOM \\\midrule
                      Pendulum &      & $10^{-2}$ & $0.995$ & $42.0$ \\
         (\gls{nndm})     & Ours & $10^{-3}$ & $0.995$ & $42.0$ \\
         - 2 layers    &      & $10^{-4}$ & $0.995$ & $42.3$ \\\cmidrule{2-5}
         - 64 neurons  &      & $10^{-2}$ & $0.903$ & $16.5$ \\
         - 480 regions & SAA  & $10^{-3}$ & $0.903$ & $47.3$ \\
                       &      & $10^{-4}$ & - & OOM \\
        \bottomrule
    \end{tabularx}}
\end{table}

To compare against state-of-the-art, we report in Table \ref{tab:results_table_saa} the certified safety by using our method and using \gls{saa}, also as the mean over 100 trials. 
Because of the greater sample complexity of \gls{saa} having $\beta=10^{-9}$ is not feasible. Hence, we compare for multiple, higher values of $\beta$ to both make it tractable and find the limits of \gls{saa}.
From the table, we observe that for our method certifying for orders of magnitude larger confidence ($10^{-2}$ to $10^{-4}$) negligibly increases the computation time ($7\%$ for the 1D linear system and $0.7\%$ for the pendulum model) and achieves similar levels of certified safety.
In contrast, for \gls{saa}, going from $10^{-2}$ to $10^{-3}$ increases the computation time between 2.8x and 13x. Furthermore, for a confidence with $\beta = 10^{-4}$, the amount of memory required exceeds the 16GB available.
The achieved level of certified safety is also marginally better with our proposed method ($0.512$ vs $0.506$ for the 1D linear system and $0.995$ vs $0.903$ for the pendulum model).

\section{Conclusion}
\label{sec:discussion}
We have presented a novel data-driven method to synthesise \acrlong{pwa} \acrfull{sbf} based on a novel inner-approximation, which relies on the scenario theory. Our approach employs \acrlong{lbp} and stochastic approximations to guarantee that the search for a barrier reduces to solving a \acrlong{lp} problem, which can be solved efficiently using the convex hull over the noise samples and spatial indexing for faster searching. 
As with any method, ours comes with limitations. In fact, in our approach we assume that the noise is additive and that we have \gls{iid} full measurements of the state of the system. 
In particular, we rely on the additivity of the noise to achieve efficient algorithms in practice. How to extend our approach to non-additive noise represents an important open question. 
Another direction for future research to improve scalability is to consider more complex piece-wise templates for parameterising a \gls{sbf}, e.g. piece-wise quadratic barrier candidates. Using more complex templates may require fewer pieces and thus improve synthesis performance. For more complex templates, one challenge is that strong duality, on which we have relied heavily, does not necessarily hold. 
Another direction of future work is to treat the safe control synthesis problem within the proposed framework, where a challenge is that the optimisation problem easily becomes bilinear.

\bibliographystyle{plain}
\bibliography{library.bib,licio-ref.bib}

\cleardoublepage
\appendix

\section{Proof and auxiliary results}
\label{sec:appendix_proofs}

\subsection{Proof of Proposition \ref{prop:main_result_robust_LP}}
\label{subsec:appendix_proof_robust_lp_duality}

\begin{proof}
    Let the sets $\mathcal{Z}$ and 
    \begin{align}
        \mathcal{Z} &= \{ z : (Az + a)^\top x \leq Bz + b, \text{ for all } x \in \polyhedron \},\\
        \mathcal{Z}'' &= \{ (z,\dualvariable) : h^\top \dualvariable \leq Bz + b,~ H^\top \dualvariable = Az + a,~\dualvariable \geq 0 \},
    \end{align}
    be the feasible sets of \eqref{eq:robust_LP} and \eqref{eq:dual_robust_LP_formulation}, respectively. Our goal is to show that $\mathcal{Z} = \mathrm{proj}_z(\mathcal{Z}'') = \mathcal{Z}'$, where 
    \[
        \mathcal{Z}' = \{ z : \exists \dualvariable \in \mathbb{R}^p_{\geq 0}, ~h^\top \dualvariable \leq Bz + b,~ H^\top \dualvariable = (Az + a) \}.
    \] is the projection of the set $\mathcal{Z}''$ onto its first $d$ coordinates. At the core of this result is the strong duality \cite[Section 5.2.1]{BV:04} between the linear programs
    \[
    \begin{aligned}
        \maximise_x & \quad (Az + a)^\top x \nonumber  \\
        \subjectto & \quad H x \leq h, 
    \end{aligned} \qquad
    \begin{aligned}
         \min_\dualvariable &\quad \dualvariable^\top h \nonumber\\
         \subjectto &\quad H^\top \dualvariable = (Az + a), \quad \dualvariable \geq 0,
    \end{aligned}
    \]
    which states that for any element $x_{\star}$ in the optimal set of the maximisation problem there exist a $\dualvariable_{\star}$ in the optimal set of the minimisation problem such that $(Az + a)^\top x_{\star} = \dualvariable_{\star}^\top h$ and $\dualvariable_{\star}^\top (H^\top x_{\star} - h) = 0.$ Vice-versa, for all $\dualvariable_{\star}$ in the optimal set of the minimisation problem there exists a $x_{\star}$ in the optimal set of the maximisation problem such that similar conclusions hold.

    Pick any element $\bar{z} \in \mathcal{Z}$. Let $\bar{x}$ be such that $(A\bar{z} + a)^\top \bar{x} = \sup_{x \in \polyhedron} a(\bar{z})^\top x$. Hence, by strong duality, there exists a $\bar{\dualvariable} \geq 0$ with $H^\top \bar{\dualvariable} = (A\bar{z} + a)$ such that $\bar{\dualvariable}^\top h = (A\bar{z} + a)^\top \bar{x} \leq B\bar{z} + b$. In other words, there is a $\bar{\dualvariable}$ such that $(\bar{z},\bar{\dualvariable}) \in \mathcal{Z}''$, which implies $\mathcal{Z} \subseteq \mathcal{Z}'$. 

    For the other direction, consider a tuple $(\bar{z},\bar{\dualvariable}) \in \mathcal{Z}''$. For this given $\bar{z}$, pick 
    \[
        \dualvariable_{\star}(\bar{z}) \in \argmin_{H^\top \dualvariable = A\bar{z} + a,~\dualvariable \geq 0} \dualvariable^\top h,
    \]
    we notice that we have that $\bar{z} \in \mathcal{Z}'$ and
    \[
    \sup_{x \in \polyhedron} (A\bar{z} + a)^\top x = \dualvariable_{\star}(\bar{z})^\top h \leq \bar{\dualvariable}^\top h \leq B\bar{z} + b.
    \]
    The right-most inequality follows from feasibility of $(\bar{z},\bar{\dualvariable})$, the middle inequality by our choice of $\dualvariable_{\star}(\bar{z})$, the left-most equality by strong duality of linear programming. Then we conclude that $\bar{z} \in \mathcal{Z}$, thus concluding the proof of the proposition.
\end{proof}

\subsection{Proof of Proposition \ref{prop:continuous_function_as_uncertain_pwa}}
\label{subsec:appendix_proof_existence_linear_bounds}

\begin{proof}
    Fix a dimension $j \in \{0, \ldots, n\}$ and a convex region $P_i$. Then due to the mean value theorem for generalised gradients \cite{CLARKE198152}, there exists two hyperplanes $(\underline{A_i}x + \underline{b_i})_j$, $(\xoverline{A_i}x + \xoverline{b_i})_j$ such that it holds
    \[
        (\underline{A_i}x + \underline{b_i})_j \leq f(x)_j \leq (\xoverline{A_i}x + \xoverline{b_i})_j
    \]
    for all $x \in \region_i$. Combining all dimensions, it holds that 
    \[
        \underline{A_i}x + \underline{b_i} \leq f(x) \leq (\xoverline{A_i}x + \xoverline{b_i})
    \]
    for all $x \in \region_i$ for each region $\region_i$. As a result, by the definition of $\hat{f}$ it holds that $f(x) \in \{ \hat{f}(x, \alpha)  : \alpha \in [0, 1] \}$ for all $x \in \statespace$ concluding the proof.
\end{proof}

\subsection{Measurability issue of Theorem \ref{theorem:barrier-ccp}}
\label{subsec:appendix_proof_measurability}

We need to show that set 
\begin{equation}
    \begin{aligned}
        E = \{\uncertaintyelement \in \uncertaintyspace:\,& B(f(x) + \noise(\uncertaintyelement)) + \buffervar \leq B(x) + \cmartingale, \\ & \text{for all } x \in X_s\}, 
    \end{aligned}
\end{equation}
is Borel measurable. First of all, notice that due to the fact that $\{\region_1,\ldots,\region_\numberregions\}$ is a finite partition of $X_s$, we have that
\begin{equation}
E = \bigcap_{\region_i \in X_s}\{  \uncertaintyelement \in \uncertaintyspace: \sup_{x\in \region_i} B(f(x)+\noise(\uncertaintyelement)) -B_i(x) \leq \cmartingale - \buffervar\}.
\end{equation}
As finite intersections of measurable sets are still measurable and a measurable function maps measurable sets into measurable sets, it is enough to show that each set 
$$E_i= \{  \noise \in \mathbb{R}^n : \sup_{x\in \region_i} B(f(x)+\noise) -B_i(x) \leq \cmartingale - \buffervar\} $$
is measurable. 
In order to do that note that when restricted to $\region_i$, $B_i$ is a linear function, while by construction $B(x)$ is upper semi-continuous. Furthermore, as composition of an upper semi-continuous function with a continuous one is still upper semi-continuous, we have that both $B(f(x)+\noise)$ and $B(f(x)+\noise) -B_i(x)$ are upper semi-continuous functions.
Consequently, by Proposition 7.32 in \cite{bertsekas2004stochastic} we have that 
$g_i(\noise)= \sup_{x\in \region_i} B(f(x)+\noise) -B_i(x) $ is upper semi-continuous. As $g_i(\noise)$ is upper semi-continuous, hence Borel measurable, we have that set $E_i$ is Borel measurable, thus concluding the proof.

\subsection{Proof for Proposition \ref{prop:finite_cmartingale_constraints}}
\begin{proof}
    Start by fixing $\alpha \in [0, 1]$, the pair $(i, j) \in I_\safeset \times I$, and the noise sample $\uncertaintyelement \in \sampleset$. Then the constraint $B(\hat{f}(x, \alpha) + \noise(\uncertaintyelement), \theta) + \buffervar \leq B(x, \theta) + \cmartingale$ for all $x \in \safeset$ can be rewritten as
    \[
        B_j(A_i(\alpha)x + b_i(\alpha) + \noise(\uncertaintyelement), \theta) + \buffervar \leq B_i(x, \theta) + \cmartingale, 
    \]
    for all $x \in \xoverline{Q}_{ij}(\uncertaintyelement)$.
    Using Prop. \ref{prop:main_result_robust_LP}, we rewrite this constraint into a dual formulation so that we can solve it in a lifted space. Thus, they become the following
    \begin{align*}
        h_{ij\uncertaintyelement}^\top \lambda_{ij\uncertaintyelement} &\leq \barrierc_i - \barrierc_j - \barrierl_j^\top (b_i(\alpha) + \noise(\uncertaintyelement)) + c - \buffervar, \\
        H^\top_{ij\uncertaintyelement} \lambda_{ij\uncertaintyelement} &= {A_i(\alpha)}^{\top} \barrierl_j - \barrierl_i,
    \end{align*}
    where $\lambda_{ij\uncertaintyelement}$ is a non-negative dual variable.
    Since both constraints are affine in $\alpha$, they will hold for all $\alpha \in [0, 1]$ if and only if they hold for the vertices, that is, $\alpha \in \{0, 1\}$. We conclude the proof by observing that for each noise sample $\uncertaintyelement \in \sampleset$, the union of $\xoverline{Q}_{ij}(\uncertaintyelement)$ for all $(i, j) \in I_\safeset \times I$ is a superset of $\safeset$.
\end{proof}

\section{Algorithm for computing over-approximation of $Q_{ij}(\omega)$}
\label{sec:appendix_algorithm_bisection}

Our approach to computing an over-approximation of $Q_{ij}(\omega)$ is detailed in Alg. \ref{alg:bisection_region}, where we start with a hyperrectangle $\xoverline{Q}_{ij}(\uncertaintyelement) = \region_i = \{ x \in \mathbb{R}^m : l \leq x \leq u \}$, with $\leq$ interpreted element-wise. 
The idea to reduce the size of $\xoverline{Q}_{ij}(\uncertaintyelement)$ is to increase $l$ (Line \ref{lst:bisection_line_lower_start}-\ref{lst:bisection_line_lower_end}) and decrease $u$ (Line \ref{lst:bisection_line_upper_start}-\ref{lst:bisection_line_upper_end}) while maintaining the over-approximation of $Q_{ij}(\uncertaintyelement)$.
Treating the axes sequentially (Line \ref{lst:bisection_line_axis_loop}), denoting the current axis by $k$, we use bisection first to find the largest $l_k$ and then the smallest $u_k$ such that 
$\xoverline{Q}_{ij}(\uncertaintyelement) \cap \{x : l_k \leq x_k \leq u_k\}$ is an over-approximation of $Q_{ij}(\uncertaintyelement)$. Then we can replace $\xoverline{Q}_{ij}(\uncertaintyelement)$ with $\xoverline{Q}_{ij}(\uncertaintyelement) \cap \{x : l_k \leq x_k \leq u_k\}$ for the next axis (Line \ref{lst:bisection_line_axis_replace}).

To perform the bisection for increasing $l_k$, we compute the midpoint $c_k^l$ between $l_k$ and $u_k$ (Line \ref{lst:bisection_lower_midpoint}), denoting them $l^l_k $ and $u^l_k$ respectively, and test if $ Q_{ij}(\uncertaintyelement) \subset \xoverline{Q}_{ij}(\uncertaintyelement) \cap \{x : c^l_k \leq x_k\}$. If true, then we may let $l^l_k = c_k^l$ (Line \ref{lst:bisection_line_select_lower}), and if not, let $u^l_k = c_k^l$ (Line \ref{lst:bisection_line_select_upper}). This procedure repeats for a fixed number of iterations and is performed mutatis mutandis to decrease $u_k$.
The question remains how to check if $ Q_{ij}(\uncertaintyelement) \subset \xoverline{Q}_{ij}(\uncertaintyelement) \cap \{x : c^l_k \leq x_k\}$, since $Q_{ij}(\uncertaintyelement)$ is unknown.
To this end, recall that $Q_{ij}(\uncertaintyelement)$ is the subset of region $\region_{i}$ that under a realisation of the noise $\uncertaintyelement$ reaches region $\region_j$ in one time step. Thus, if the image of the other subregion $\im(\xoverline{Q}_{ij}(\uncertaintyelement) \cap \{x : x_k \leq c^l_k \}, \uncertaintyelement)$ under the realisation of the noise $\uncertaintyelement$ does not intersect $\region_{j}$ (Line \ref{lst:bisection_lower_intersect_check}), then it necessarily holds that $ Q_{ij}(\uncertaintyelement) \subset \xoverline{Q}_{ij}(\uncertaintyelement) \cap \{x : c^l_k \leq x_k\}$.

\begin{algorithm}
\caption{Bisection-based algorithm for computing a subset $\xoverline{Q}_{ij}(\uncertaintyelement)$ of region $\region_i$ as an over-approximation of $Q_{ij}(\uncertaintyelement)$.}\label{alg:bisection_region}
\begin{algorithmic}[1]
\Function{PolyPreimage}{$\region_i$, $\region_j$, $\hat{f}$, $\uncertaintyelement$, $t$}
\State $\xoverline{Q}_{ij}(\uncertaintyelement) \gets \region_i$
\For{$k \gets 1 $ to $m$} \Comment{For each axis} \label{lst:bisection_line_axis_loop}
    \State
    \State $l^l_k, u^l_k \gets l_k, u_k$ \Comment{Increase lower bound} \label{lst:bisection_line_lower_start}
    \For{$s \gets 1$ to $t$}
        \State $c^l_k \gets \frac{l^l_k + u^l_k}{2}$ \label{lst:bisection_lower_midpoint}
        \State $\xoverline{Q}_{ij}(\uncertaintyelement)' \gets \xoverline{Q}_{ij}(\uncertaintyelement) \cap \{ x : x_k \leq c^l_k \}$
        \If{$\im_i(\xoverline{Q}_{ij}(\uncertaintyelement)', \uncertaintyelement) \cap \region_j = \emptyset $} \label{lst:bisection_lower_intersect_check}
            \State $l^l_k \gets c^l_k$ \label{lst:bisection_line_select_lower}
        \Else
            \State $u^l_k \gets c^l_k$ \label{lst:bisection_line_select_upper}
        \EndIf
    \EndFor \label{lst:bisection_line_lower_end}
    \State
    \State $l^u_k, u^u_k \gets l_k, u_k$ \Comment{Decrease upper bound} \label{lst:bisection_line_upper_start}
    \For{$s \gets 1$ to $t$}
        \State $c^u_k \gets \frac{l^u_k + u^u_k}{2}$
        \State $\xoverline{Q}_{ij}(\uncertaintyelement)' \gets \xoverline{Q}_{ij}(\uncertaintyelement) \cap \{ x : x_k \geq c^u_k \}$
        \If{$\im_i(\xoverline{Q}_{ij}(\uncertaintyelement)', \uncertaintyelement) \cap \region_j = \emptyset $}
            \State $u^u_k \gets c^u_k$
        \Else
            \State $l^u_k \gets c^u_k$
        \EndIf
    \EndFor \label{lst:bisection_line_upper_end}
    \State
    \State $\xoverline{Q}_{ij}(\uncertaintyelement) \gets \xoverline{Q}_{ij}(\uncertaintyelement) \cap \{x : l^l_k \leq x_k \leq u^u_k \}$ \label{lst:bisection_line_axis_replace}
\EndFor
\State \textbf{return} $\xoverline{Q}_{ij}(\uncertaintyelement)$ \Comment{It holds that $Q_{ij}(\uncertaintyelement) \subset \xoverline{Q}_{ij}(\uncertaintyelement)$}
\EndFunction
\end{algorithmic}
\end{algorithm}

\end{document}